\documentclass[conference]{IEEEtran}
\IEEEoverridecommandlockouts

\usepackage{cite}
\usepackage{amsmath,amssymb,amsfonts}
\usepackage{algorithmic}
\usepackage{graphicx}
\usepackage{textcomp}
\usepackage{xcolor}

\usepackage{breqn}

\usepackage{amsthm}
\usepackage{algorithm}
\usepackage{mathtools}
\usepackage{multirow}
\usepackage{color}
\usepackage{subfigure}
\usepackage{placeins}
\usepackage{hyperref}

\newtheorem{theorem}{Theorem}

\newtheorem{cor}{Corollary}
\newtheorem{definition}{Definition}
\newtheorem{example}{Example}

\newcommand{\X}{{\mathcal{X}}}
\newcommand{\U}{{\mathcal{U}}}
\newcommand{\K}{{\mathcal{K}}}
\newcommand{\cS}{{\mathcal{S}}}

\def\BibTeX{{\rm B\kern-.05em{\sc i\kern-.025em b}\kern-.08em
		T\kern-.1667em\lower.7ex\hbox{E}\kern-.125emX}}
\begin{document}
	
	\title{A Robust Optimization Approach for Regenerator Placement in Fault-Tolerant Networks Under Discrete Cost Uncertainty\\
	}
	
	\author{
		\IEEEauthorblockN{Mohammad Khosravi and Setareh Maghsudi}
		\IEEEauthorblockA{Learning Technical Systems, Ruhr University of Bochum, Bochum, Germany}
	}

	\maketitle
	\begin{abstract}
		We focus on robust, survivable communication networks, where network links and nodes are affected by an uncertainty set. In this sense, any network links might fail. Besides, a signal can only travel a maximum distance before its quality falls below a certain threshold, necessitating its regeneration by regenerators installed at network nodes. In addition, the price of installing and maintaining regenerators belongs to a discrete uncertainty set. Robust optimization seeks a solution with guaranteed performance against all scenarios modeled in an uncertainty set. Thus, the problem is to find a subset of nodes with minimum cost for the placement of the regenerator, ensuring that all nodes can communicate even if a subset of network links fails. To solve the problem optimally, we propose two solution approaches, including one flow-based and one cut-based integer programming formulation, as well as their iterative exact method. Our theoretical and experimental results show the effectiveness of our methods.
	\end{abstract}
	\section{Introduction}
	\label{sec:introduction}
	Modern societies have been operating and expanding through networks in all forms ranging from global communications to supporting critical services. As the impact of networks increases in our lives, the importance of ensuring their integrated functionalities and resilience becomes undeniable. \textbf{Network resilience} is defined as the ability of a network to remain operational under critical circumstances. In other words, the resilience of a network refers to its ability to maintain its functionality in the face of potential disruptions so that it can survive, recover, and adapt to disruptions. Thus, survivability is one of the most critical steps of modern network design, as it involves strategies to ensure that networks can operate during and after specific critical circumstances. 
	
	We focus on the resilient design of communication networks based on the maximum distance $d_{max}$ that data/signals can travel between network nodes without a significant degradation in quality. That is, the quality of a signal falls below an acceptable threshold if it travels over a path with a length greater than $d_{\max}$. Besides, we allow the network to be faulty, that is, the transmission links, represented by edges in the network graph, are subject to failure. Then, one can manage faults by reserving backup resources, known as the \textit{protection scheme}, or by finding available spare resources in real time, referred to as the \textit{restoration scheme} \cite{li2020branch}. 
	
	To navigate into a solution, we first note that in such a network, a signal traversing longer than $d_{\max}$ must be regenerated, using a regenerator deployed at some network nodes. However, regenerators are expensive; as such, a large body of research focuses on minimizing their usage while satisfying communication requirements. That challenge yields the \textit{regenerator location problem} (RLP) (also known as the \textit{regenerator placement problem} (RPP)), which seeks the best, that is, the most efficient location(s) for deploying regenerator(s) in the network. This seminal version of RLP appeared in \cite{sen2008sparse}: Given a network and $d_{max}$, the objective is to find the fewest number of network nodes to place regenerators so that every node pair can communicate through a path in which every subpath longer than $d_{max}$ has a regenerator on all its internal nodes \cite{sen2008sparse}. To also ensure that the network is fault-tolerant, \cite{rahman2014optimal} and \cite{li2020branch} formulate and solve the fault-tolerant regenerator location problem (FTRLP).
	
	Often, the cost of installing and maintaining regenerator equipment varies among nodes, for example, due to their geographical sites. That gives rise to the weighted version of FTRLP \cite{chen2010regenerator}. In this paper, we define the \textbf{robust version of weighted FTRLP}. More precisely, we allow uncertain deployment and maintenance price for the regenerator equipment over different network nodes, assuming that all values belong to a specific uncertainty set, as described in the following definition. 
	\begin{definition} 
		Consider a communication network where any subset of links might fail, and the maximum distance for signal transmission is $d_{\max}$. The robust FTRLP aims to determine the minimum-cost subset of network nodes for regenerator placement. The placement must ensure a path for every pair of network nodes, such that no path segment without any internal regenerator exceeds the maximum distance $d_{\max}$ for any possible link failure.
	\end{definition}
	\subsection{Application Areas}
	\label{subsec:application}
	\subsubsection*{Optical Networks}
	\label{subsec:optical}
	Optical networks are essential to support the enormous data demand, as they enable high-speed transmission over long distances by offering high bandwidth and low latency. However, maintaining signal quality over long distances requires using optical regenerators or amplifiers, which restore signal strength and prevent degradation. Thus, RLP is crucial in optical network design, as the location of regenerators directly impacts overall network performance.
	\subsubsection*{Corporate Communication Networks}
	\label{subsec:corporate}
	A modern corporate network is a communication system in which all operational parts of a company merge into a single team (regardless of their location) to ensure the effective functioning of business processes. The fault-tolerant regenerator location problem directly connects to ensuring the resilience of corporate communication networks, where reliable network infrastructure is essential for both internal communication and secure data transfer over long distances. 
	\subsubsection*{Relay-Aided Communication Networks}
	\label{subsec:relay}
	By dynamically extending network reach, relays remarkably contribute to network resilience, ensuring communication remains operational when some infrastructures fail. This is directly related to the FTRLP, as both require strategic placement of critical equipment (relays) to enhance network performance. In both cases, optimizing the placement ensures safe communication, reducing bottlenecks and supporting the network survivability under critical conditions.
	\subsection{Related Works}
	\label{subsec:relatedworks}
	In \cite{chen2010regenerator,sen2008sparse}, the complexity of the RLP is proved to be NP-complete. As RLP is a special case of the FTRLP where all links remain stable, FTRLP is also NP-complete.
	
	Reference \cite{yildiz2015regenerator} studies two types of network resilience under partial and full survivability conditions. The former refers to the resilience against regenerator failure (RLPRF) while the latter corresponds to the resilience against node failure. The results show that RLPNF is slightly harder to solve than RLPRF in terms of solution time. While the RLP in flexible networks is studied in \cite{yildiz2017regenerator}, the authors in \cite{chen2015generalized} investigated a generalized version called GRLP, where the objective is to enable communication between a given subset of nodes by installing regenerators on another specific subset of network nodes.
	
	Several studies investigate regenerator placement as part of joint optimization problems in elastic optical networks. In \cite{gonzalez2021regeneration}, a binary integer programming model is proposed to jointly solve regeneration placement, routing, modulation level, and spectrum assignment (RP--RMLSA), showing that optimal solutions often concentrate regeneration at a single node depending on the network topology and cost structure. Similarly, \cite{wang2015impact} develops a link-based MILP formulation to evaluate the impact of signal regeneration, wavelength conversion, and modulation conversion on spectrum usage in EONs, and proposes a recursive solution method to balance optimality and computational complexity. While these works provide important insights into joint design and resource allocation in elastic optical networks, they consider deterministic settings and do not address fault tolerance or robustness under failures and uncertainty.
	
	Only a small number of papers investigate FTRLP. In \cite{rahman2014optimal}, the authors introduce an integer programming formulation with an exponential number of set-covering inequalities. Reference \cite{li2020branch} develops a flow-based formulation for all possible fault scenarios which depends on the number of edges in the network. Both papers use branching approaches to solve the formulated problems: The former uses a branch-and-cut (B$\&$C) approach and the latter implements a branch-and-Benders-cut (BBC) approach inspired by Benders decomposition \cite{benders2005partitioning}. Their numerical results show that the solution provided by \cite{li2020branch} is considerably better than that by \cite{rahman2014optimal}. Regardless, by increasing the size of instances, both methods become provably time-consuming.
	
	Considering the RLP under uncertainty, there exists even less studies in the literature. For instance, \cite{yan2018robust} considers a greenfield scenario, where all source-destination pairs are known but their transmission rates change over time. The authors then present a regenerator assignment strategy based on the probability of requiring a regenerator at a given node. 
	\subsection{Our Contribution}
	\label{subsec:contribution}
	In summary, the contributions of this paper are as follows.
	\begin{itemize}
		\item We consider fault-tolerant networks, where the network edges are uncertain and may fail. While we focus on cases where only one edge can fail, our main contribution lies in the extension of our method where multiple edges may fail.
		\item We concentrate on node-weighted networks, where the cost of installing and maintaining a regenerator varies across nodes. As such, our main contribution is the introduction of a robust version of the problem under discrete uncertainty sets. The uncertainty set includes the potential realizations of prices over each node.
		\item We introduce a constraint called \emph{Full Recovery of Equality} that allows us to generate a single transformation graph to implement our integer programming formulations which decreases the pre-processing time needed.
		\item We propose IP-FB and IP-CB, two novel mathematical formulations to solve the problem exactly. Moreover, we propose an iterative method that is efficient for large-scale instances.
	\end{itemize}
	\subsection{Organization}
	\label{subsec:organization}
	The remainder of this paper is organized as follows. Section~\ref{subsec:notations} provides an overview of the notation used throughout the paper. In Section~\ref{sec:problem-statement}, we formally define the RLP, introduce the communication graph, and present problems equivalent to the RLP. Section~\ref{sec:robust-counterpart} investigates the robust version of the RLP, examines its computational complexity, and introduces its equivalent problem. In Section~\ref{sec:mathematical-models}, we outline the steps for effectively solving the robust problem by proposing a preprocessing approach, followed by two mathematical formulations and an iterative method. Section~\ref{sec:experiments} summarizes the results of our numerical experiments. Finally, Section~\ref{sec:conclusion} concludes the paper. Additionally, we extend our constraints and solution approaches to the case of multiple edge failures in Appendices~\ref{app:extension-fre}, and \ref{app:extension}.
	
	\subsection{Notations}
	\label{subsec:notations}
	In this paper, bold lowercase letters are used to denote vectors. We define the notation $[m]=\{1,2,\ldots,m\}$ to represent the set of the first $m$ positive integers. The cost vector is denoted by $\pmb{c}$, the underlying problem structure is represented by $\X$, and the set of given uncertain parameters is denoted by $\U$. We also provide a complete list of widely used abbreviations in this paper summarized in Table~\ref{table:abbreviations}.
	\vspace{-0.2cm}
	\begin{table}[htbp]
		\begin{center}
			\caption{List of abbreviations}
			\begin{tabular}{|c|c|}
				\hline
				Abbreviation & Full Term \\
				\hline
				RLP & Regenerator Location Problem \\
				FTRLP & Fault-Tolerant RLP \\
				RFTRLP & Robust Fault-Tolerant RLP \\
				MLSTP & Maximum Leaf Spanning Tree Problem\\
				MCDSP & Minimum Connected Dominating Set Problem \\
				FTMCDSP & Fault-Tolerant MCDSP \\
				RFTMCDSP & Robust Fault-Tolerant MCDSP \\
				FRE & Full Recovery of Equality \\
				\hline
			\end{tabular}\label{table:abbreviations}
		\end{center}
	\end{table}
	
	\section{Problem Statement}
	\label{sec:problem-statement}
	Consider a network $G=$($V,E,D$), where $V$ is the set of nodes with $\lvert V\rvert=n$, $E$ is the set of edges with $\lvert E\rvert=m$, and $D$ is the corresponding distance matrix of edges. The maximum distance is $d_{max}$. Any travel distance exceeding $d_{max}$ implies that the signal must be regenerated at one or more intermediate nodes along the path.
	
	In addition, it must be noticed that edges with length more than $d_{max}$ in networks (if any) cannot be used for signal transmission, as doing so would result in signal failure. Hence, from this point onward, we assume that all edges of networks have length less than or equal to $d_{max}$. For a visual illustration of the problem, refer to Figure~\ref{fig:example-network}.
	\vspace{-0.2cm}
	\begin{figure}[htbp]
		\begin{center}
			\includegraphics[width=0.45\textwidth]{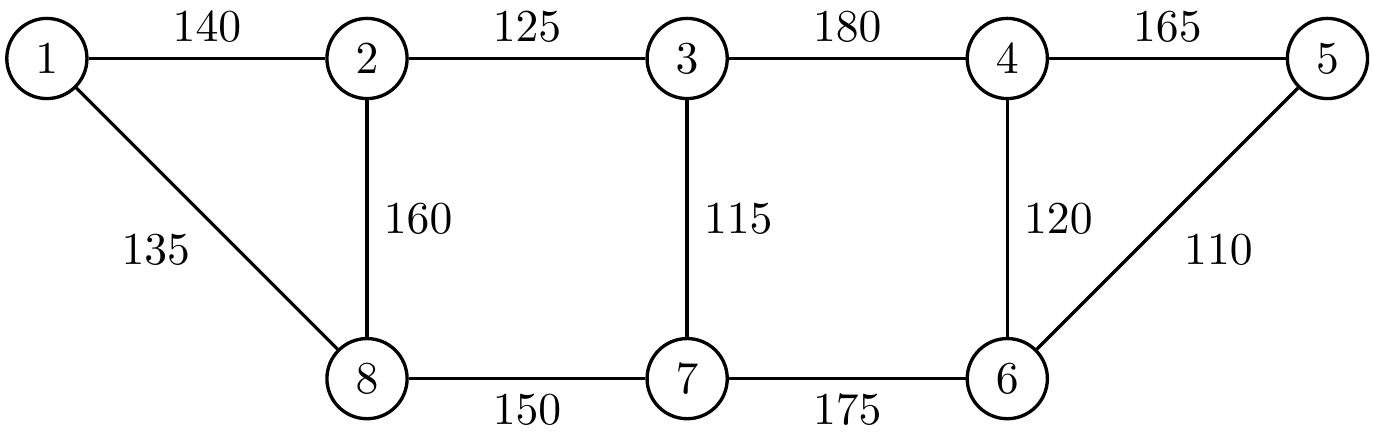}
		\end{center}
		\caption{Example of an eight-node network with $d_{\max}=200$. 
			The optimal solutions of the RLP are either nodes $\{2,3,4\}$ or nodes $\{6,7,8\}$. Considering $\{2,3,4\}$ as the solution, nodes~1 and~5 can communicate via the path $1\leftrightarrow2\leftrightarrow3\leftrightarrow4\leftrightarrow5$. Using this placement, a valid communication path can be established for all non-adjacent node pairs in the network.}
		\label{fig:example-network}
	\end{figure}
	\subsection{The Communication Graph}
	\label{subsec:communicationg-graph}
	We now describe the graph transformation approach that we use on the RLP, using a similar method introduced in \cite{sen2008sparse}. This transformation method is necessary for the implementation of both our exact solution and iterative methods.
	
	Given graph $G$ and the maximum distance of $d_{max}$, we first apply the all-pair shortest path algorithm on graph $G$. We keep an edge $(i,j)$ only if the shortest distance between node $i$ and $j$ in $G$ is less than or equal to $d_{max}$, and all such edges form $E_M$. We refer to the remaining graph as the \textit{transformed graph} $M=$($V,E_M$). Clearly, if $M$ is a complete graph, all pair of nodes in the RLP instance can communicate without any regenerator; Otherwise, communication between any pair of nodes that is not connected by an edge in $M$, called \textit{NDC pairs}, is only possible through a path with regenerators on all its internal nodes. Therefore, the RLP can be equivalently restated on graph $M$ as follows.
	\begin{definition}[RLP on M]
		Determine a subset of nodes with minimum cardinality for regenerator equipment such that every NDC pair of nodes in $M$ can communicate through a path with regenerators located on its all internal nodes.
	\end{definition}
	\subsection{Equivalent Problems to RLP}
	\label{subsec:equivalent-problems-rlp}
	In this section, we introduce problems that are equivalent to the RLP, which are essential for developing the mathematical formulations presented in Section~\ref{sec:mathematical-models}.
	\begin{definition}[Dominating Set Problem]
		Given an undirected graph $G^{\prime}=$($V^{\prime},E^{\prime}$), a subset $T\subseteq V^{\prime}$ is a dominating set if every vertex $v^{\prime}\in V^{\prime}\backslash T$ is joined to at least one node of $T$ by an edge in $E^{\prime}$.
	\end{definition}
	A connected dominating set is a dominating set $T$ such that the subgraph $G^{\prime}_T=$($T$,$E^{\prime}_T$) is connected, where $E^{\prime}_T$ is the set of edges of $E^{\prime}$ with both endpoints in $T$. Moreover, the problem of finding a connected dominating set with the minimum cardinality is referred to as the \textit{minimum connected dominating set problem} (MCDSP).
	It is proved in \cite{chen2010regenerator} that the RLP is equivalent to the \textit{maximum leaf spanning tree problem} (MLSTP). Moreover, it is easy to see, as observed in \cite{gendron2014benders} and \cite{lucena2010reformulations}, that the MLSTP is equivalent to the MCDSP and thus all three problems of RLP, MLSTP and MCDSP are equivalent. However, it must be noticed that the equivalency between these problems only holds on the communication graph $M$ and not the original graph $G$. Consequently, an equivalent definition of the RLP (based on the concept of MCDSP that is used to form our mathematical formulations in Section~\ref{sec:mathematical-models}) on $M$ can be described as follows.
	\begin{definition}[Equivalent RLP on $M$]
		Determine a subset of nodes with minimum cardinality for regenerator placement such that the subgraph induced by this subset of nodes is connected, and every node which is not included in this subset is connected to at least one of its nodes.
	\end{definition}
	\section{Robust Counterpart}
	\label{sec:robust-counterpart}
	Now we focus on the robust version of the RLP affected by two types of uncertainty. The first source of uncertainty is the cost of deploying a regenerator at each node, which varies over the network. We assume that all costs are given through a discrete uncertainty set of the form
	\[\U_V = \{\pmb{c}^1,\ldots,\pmb{c}^N\}\]
	where $\pmb{c}^i\in \mathbb{R}^n_+$ for all $i\in[N]$. Therefore, the robust counterpart of the RLP yields
	\[\min_{\pmb{x}\in\X} \max_{\pmb{c}\in\U_V} \; \pmb{c}^\intercal \pmb{x},\] 
	where $\X$ returns the structure of the RLP. 
	
	For both the nominal and robust RLP under $\U_V$, we find a binary solution $\pmb{x}$ that represents the regenerators locations. However, the objective of the nominal problem is to minimize the cardinality of the solution, whereas for the robust version under $\U_V$, the desire is to minimize the total cost of deploying regenerators. In other words, if the nominal version yields multiple subsets of nodes with the same cardinality, the choice among them is inconsequential. However, in the robust counterpart under $\U_V$, different subsets of the same size may incur different costs. More importantly, the robust solution with respect to $\U_V$ might have a higher cardinality than the nominal solution but with lower cost; nevertheless, as the discrete uncertainty only affects the objective function of the problem, the state-of-the-art mathematical formulations remain applicable with modification of the objective function.
	
	The second source of uncertainty corresponds to the edges. We assume that $\Gamma$ edges might be disconnected by the nature as an adversary, or due to adversarial attacks. Thus, we can formulate this uncertainty as a new version of the budgeted interdiction uncertainty introduced in \cite{goerigk2024robust}. Thus,  it can be written as
	\[\U_E=\{\tilde{\pmb{c}}\in\{0,1\}^m \colon \sum_{i\in[m]} \tilde{c}_i e_i \le \Gamma\},\]
	where $e_i$ is an edge of the graph for all $i\in [m]$. In this sense, the uncertainty set $\U_E$ is all the possible attacks of the adversary, and it affects the structure of the RLP. Consequently, the robust fault-tolerant formulation of the RLP (RFTRLP) is given below.
	\[\min_{\pmb{x}\in\X_{\U_E}} \max_{\pmb{c}\in\U_V} \; \pmb{c}^\intercal \pmb{x}\]
	
	In words, in RFTRLP, we seek a solution that guarantees at least a backup route between all pair of nodes, which remains unaffected if $\Gamma$ edges of the network fail. In this paper, we only focus on the case where $\Gamma=1$, but our methods introduced in Section~\ref{sec:mathematical-models} can be extended to larger $\Gamma$ values (see Appendix~\ref{app:extension}).
	
	First, we need to investigate the network to ensure that the RFTRLP has a solution.
	\begin{theorem}\label{the:graph-connectivity}
		The RFTRLP for $\Gamma$ edge failure has a solution if $G$ is a ($\Gamma+1$)-edge-connected graph.
	\end{theorem}
	\begin{proof}
		If graph $G$ is a ($\Gamma+1$)-edge-connected, then there exist $\Gamma+1$ edge-disjoint paths between any pair of nodes, and thus by removing $\Gamma$ edges from the graph, at least one path remains to connect any pair of nodes in the graph. Consequently the problem has a solution.
	\end{proof}
	
	Therefore, for $\Gamma=1$, graph $G$ must be a two-edge-connected graph. Now, we can formally define the RFTRLP for a single edge failure.
	\begin{definition}
		Determine a subset of nodes with minimum cost such that there are at least two edge-disjoint paths between every NDC pair of nodes ($s$,$t$) in $G$, where there exists no sub-path in these two paths with length larger than or equal to $d_{max}$ without internal regenerators.
	\end{definition}

	\subsection{Problem Complexity}
	
	The computational complexity of the RLP has been established in \cite{chen2010regenerator} through a step-by-step reduction to the Vertex Cover problem.
	\begin{theorem}
		The regenerator location problem is NP-complete.
	\end{theorem}
	
	The robust counterpart of the combinatorial problems under discrete uncertainty set are usually significantly more challenging to solve. For instance, the nominal selection problem can be solved in polynomial time, but its robust counterpart becomes NP-hard when the number of scenarios is fixed (even when $N=2$), and strongly NP-hard and not apprximable within any constant when the number of scenarios is part of the inputs \cite{kasperski2016robust}. 
	
	Now we can prove the complexity level of the RFTRLP. Therefore, we first use a similar approach introduced in \cite{chen2010regenerator} to reduce the robust weighted version of the RLP under discrete uncertainty to the robust weighted vertex cover problem with the same uncertainty set, which is known to be NP-hard \cite{goerigk2023optimal}.
	\begin{theorem}\label{theorem:RWRLP}
		The robust weighted regenerator location problem under discrete uncertainty set is NP-hard.
	\end{theorem}
	\begin{proof}
		Robust weighted vertex cover problem (RWVCP): given an undirected graph $G=(V,E)$ and the discrete uncertainty set $\U=\{\pmb{c}^1,\ldots,\pmb{c}^N\}$, representing the node costs, select a subset $S$ of nodes with minimum costs such that each node $v \in V$ is either in $S$ or has a neighbor in $S$.\\
		We now construct the corresponding instance of RWRLP. Create a node for every $v\in V$. Create an edge between every pair of nodes in $V$. For every edge $e_j \in E$, create a pair of nodes $u_j$ and $w_j$. Connect both $u_j$ and $w_j$ to the endpoints of edge $e_j$. Set the length of all edges in the resulting graph to $d_{max}$. In addition, for all $i\in[N]$ fix the cost of all $u$ and $w$ nodes to $M^i+1$ where $M^i=\max \{c^i_1\ldots,c^i_n\}$.\\
		The question in this new graph is whether there is a feasible solution (a subset of nodes $L$ regenerator placement) to the RWRLP with a minimum sum of node costs. Observe that in the  RWRLP obtained from transforming an RWVCP, a feasible solution does not need to place a regenerator at $u_j$ and $w_j$ nodes. This is because the nodes in $V$ are fully connected in the graph corresponding to the RWRLP problem. Thus, a feasible solution with a regenerator at a $v_j$ or $w_j$ node remains feasible when the regenerator is removed. With this, it is easy to observe that an instance of RWVCP has a “yes” answer if and only if the corresponding RWRLP has a “yes” answer. Hence, the robust weighted regenerator location problem under a discrete uncertainty set is NP-hard.
	\end{proof}
	\begin{cor}\label{cor:RFTRLP}
		The robust fault-tolerant regenerator location problem under uncertainty sets $\U_V$ and $\U_E$ is NP-hard.
	\end{cor}
	\begin{proof}
		Using a similar technique to that introduced in \cite{li2020branch}, we can define a fault scenario by removing each edge form $G$ and then obtain the transformed graph $M$. Consequently, the problem becomes finding a solution for the robust weighted RLP that is valid over each transformed graph. Therefore, the complexity level of the RFTRLP can be proved.
	\end{proof}
	%
	\subsection{Equivalent Problems to RFTRLP}
	\label{subsec:equivalent-problems-rftrlp}
	Similar to the approach presented in Section~\ref{subsec:equivalent-problems-rlp}, we introduce an equivalent problem to the FTRLP that can be modeled explicitly using integer programming techniques. This problem preserves the structure and optimality properties of the FTRLP, and therefore any optimal solution obtained for the problem directly corresponds to an optimal solution of the FTRLP.
	\begin{definition}[Fault-Tolerant MCDSP]\label{def:FTMCDSP}
		A special case of MCDSP that remains feasible if an edge is removed from the network.
	\end{definition}
	
	We refer to the problem introduced in Definition~\ref{def:FTMCDSP} as FTMCDSP. However, the solution of FTMCDSP and FTRLP on $M$ might not remain feasible on $G$, due to the fact that two edge-disjoint paths on graph $M$ might not be edge-disjoint on $G$. To better understand it see Example~\ref{exp:inquality-on-m}. 
	\begin{example}\label{exp:inquality-on-m}
		Consider the network shown in Figure~\ref{fig:example-inequality}(a), where $d_{\max}=100$. From the corresponding graph $M$, depicted in Figure~\ref{fig:example-inequality}(b), it can be observed that the subgraph induced by nodes $\{2,3,4,5\}$ is two-edge-connected. Furthermore, nodes $\{1,6\}$ are each connected to exactly two nodes in the set $\{2,3,4,5\}$. As a result, the subset $\{2,3,4,5\}$ shows a feasible solution for both the FTRLP and the FTMCDSP on graph $M$. Since no solution of smaller cardinality exists, this subset is also optimal for the FTRLP.
		Now, suppose that edge $(1,2)$ is removed from the original graph $G$, and the corresponding graph $M_{G\setminus(1,2)}$ is constructed, as shown in Figure~\ref{fig:example-inequality}(c). In this case, communication between node~1 and nodes $\{2,3,4,5\}$ in the graph $G\setminus(1,2)$ is no longer possible. This occurs because, although node~1 is connected to two nodes in the optimal solution in graph $M$, the edge $(1,3)$ in $M$ corresponds to the path $1 \rightarrow 2 \rightarrow 3$ in $G$, which includes edge $(1,2)$. Consequently, once edge $(1,2)$ is removed from $G$, the edge $(1,3)$ cannot be constructed in the derived graph $M_{G\setminus(1,2)}$. This example illustrates the inequivalence between the solutions of the FTRLP and the FTMCDSP on graphs $G$ and $M$.
		\vspace{-0.2cm}
		\begin{figure}[htbp]
			\begin{center}
				\subfigure[Graph $G$]{
					\includegraphics[width=0.14\textwidth]{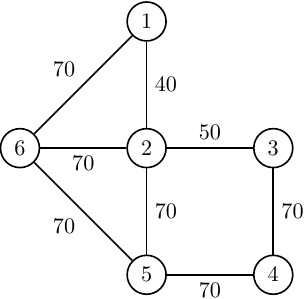}
				}\hfill
				\subfigure[Graph $M$]{
					\includegraphics[width=0.14\textwidth]{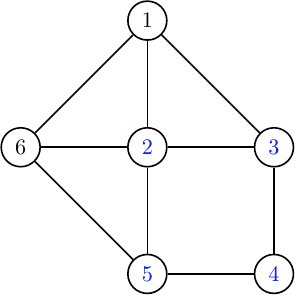}
				}\hfill
				\subfigure[Graph $M_{ G\backslash(1,2)}$]{
					\includegraphics[width=0.14\textwidth]{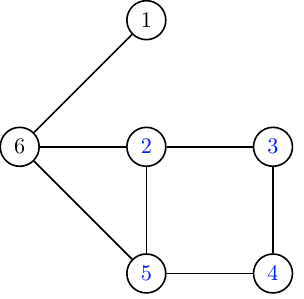}
				}
			\end{center}
			\caption{Illustration of the inequivalence between the FTRLP and FTMCDSP solutions on graphs $G$ and $M$: (a) original graph $G$, (b) corresponding graph $M$, and (c) graph $M_{G\setminus(1,2)}$ after removing edge $(1,2)$.}\label{fig:example-inequality}
		\end{figure}
	\end{example}
	Note that, in case the transformed graph $M$ remains the same as the original graph $G$, the solution of both FTRLP and FTMCDSP in $M$ remains feasible in graph $G$. In this case, the shortest path algorithm returns edges that already exist in the original graph, with possibly shorter length; That is, in construction of $M$, no new edge can be added to the original set of edges in $G$.
	
	To establish equivalence between the solutions of the FTRLP and the FTMCDSP on both graphs $G$ and $M$, an additional constraint is required. We refer to this constraint as the \emph{Full Recovery of Equality (FRE)} of the solutions on $G$ and $M$ for both the FTRLP and the FTMCDSP. To formalize this concept, we first define a new set of nodes $\overline{V}$ in the following way:
	\[
	\overline{V} = \{v \ | \ v \in V, \ \mathcal{N}_v^M > \mathcal{N}_v^G \}
	\]
	where $\mathcal{N}_v^G$ and $\mathcal{N}_v^M$ represent the number of neighbors of $v$ on $G$ and $M$, respectively.
	\begin{definition}[Full Recovery of Equality]
		Given a graph $G$ and its corresponding transformed graph $M$, the solutions of the FTRLP and the FTMCDSP on $G$ and $M$ are said to be equivalent if, for any node $v \in \overline{V}$, the number of neighbors of $v$ in $M$ that are equipped with regenerators is strictly greater than the number of equipped nodes lying on all shortest paths in $G$ that induce edges incident to $v$ in $M$.
	\end{definition}
	To justify the use of the FRE constraint, we prove in Theorem~\ref{thm:fre} that this constraint preserves the equivalence between the solutions of the FTRLP and the FTMCDSP on graphs $G$ and $M$. For further intuition on how this constraint works, Example~\ref{exp:fre} provides an illustrative scenario. In addition, an extension of the FRE constraint for cases with $\Gamma \geq 2$ is provided in Appendix~\ref{app:extension-fre}.
	\begin{theorem}\label{thm:fre}
		The solutions of the FTRLP and the FTMCDSP on graphs $G$ and $M$ are equivalent if and only if the FRE constraint holds.
	\end{theorem}
	\begin{proof}
		Let $L \subseteq V$ denote the set of nodes equipped with regenerators. Consider an arbitrary node $v \in \overline{V}$. Since $G$ is a two-edge-connected graph, node $v$ has at least two neighbors in $G$. Without loss of generality, assume that $v$ has exactly two neighbors in $G$, denoted by $u_1$ and $u_2$. Moreover, because $v \in \overline{V}$, it has additional neighbors in $M$ that do not correspond to edges in $G$. Each such neighbor of $v$ in $M$ is generated by the shortest-path algorithm and corresponds to a path in $G$ that originates at $v$ and passes through either $u_1$ or $u_2$.
		
		Sufficiency. Assume first that the FRE constraint holds. By definition, the number of neighbors of $v$ in $M$ that belong to $L$ is strictly greater than the number of nodes in $L$ lying on all shortest paths in $G$ that induce edges incident to $v$ in $M$. We now consider the possible cases.
		
		If both $u_1, u_2 \in L$, then there exist two edge-disjoint paths between $v$ and the nodes in $L$ in both graphs $G$ and $M$, and two-connectivity is immediately satisfied. Otherwise, suppose that at least one new neighbor of $v$ in $M$ is equipped with a regenerator and that the corresponding shortest path in $G$ passes through $u_1$. In this case, the FRE constraint implies that $u_2 \in L$. Consequently, there exist two edge-disjoint paths between $v$ and the nodes in $L$. Therefore, the equality between the solutions of the FTRLP and the FTMCDSP on graphs $G$ and $M$ is preserved.
		
		Necessity. Assume that the solutions of the FTRLP and the FTMCDSP on $G$ and $M$ are equal. In other words, there exist two edge-disjoint paths between $v$ and the nodes in $L$.
		
		If both $u_1, u_2 \in L$, then the FRE constraint is trivially satisfied. Otherwise, suppose that a new neighbor of $v$ in $M$ (obtained via the shortest-path algorithm and passing through $u_1$) belongs to $L$. Installing a regenerator only at $u_1$ is not sufficient in this case. Indeed, after the removal of edge $(v,u_1)$ from $G$ and the subsequent construction of the corresponding graph $M$, node $v$ remains connected to the rest of the network solely through node $u_2$. To preserve two-edge connectivity, node $u_2$ must therefore also belong to $L$. Hence, the FRE constraint necessarily holds.
	\end{proof}
	\begin{example}\label{exp:fre}
		In Figure~\ref{fig:example-inequality}, we have $\overline{V} = \{1,3\}$.
		Consequently, for node~1, the number of regenerators placed on nodes $\{2,3,6\}$ must be greater than the number placed on nodes $\{2,3\}$. Therefore, node~6 must be included in the solution subset in order to satisfy the FRE condition.
	\end{example}
	Accordingly, an equivalent definition of the RFTRLP on $M$, using the concept of FTMCDSP and FRE, is given in Definition~\ref{def:mcdsp}.
	\begin{definition}[Equivalent RFTRLP on $M$]\label{def:mcdsp}
	Determine a subset of nodes $L$ with minimum cost, such that there exist at least two edge-disjoint paths between all pair of nodes in $L$ and all nodes excluded from $L$ must be connected to at least two nodes of $L$, while satisfying the full recovery of equality of solution on $M$ and $G$.
	\end{definition}

	Finally, to clearly understand the relation between all problems we introduced in this paper see Figure~\ref{fig:flowchart}.
	\begin{figure}[htbp]
		\begin{center}
			\includegraphics[width=0.45\textwidth]{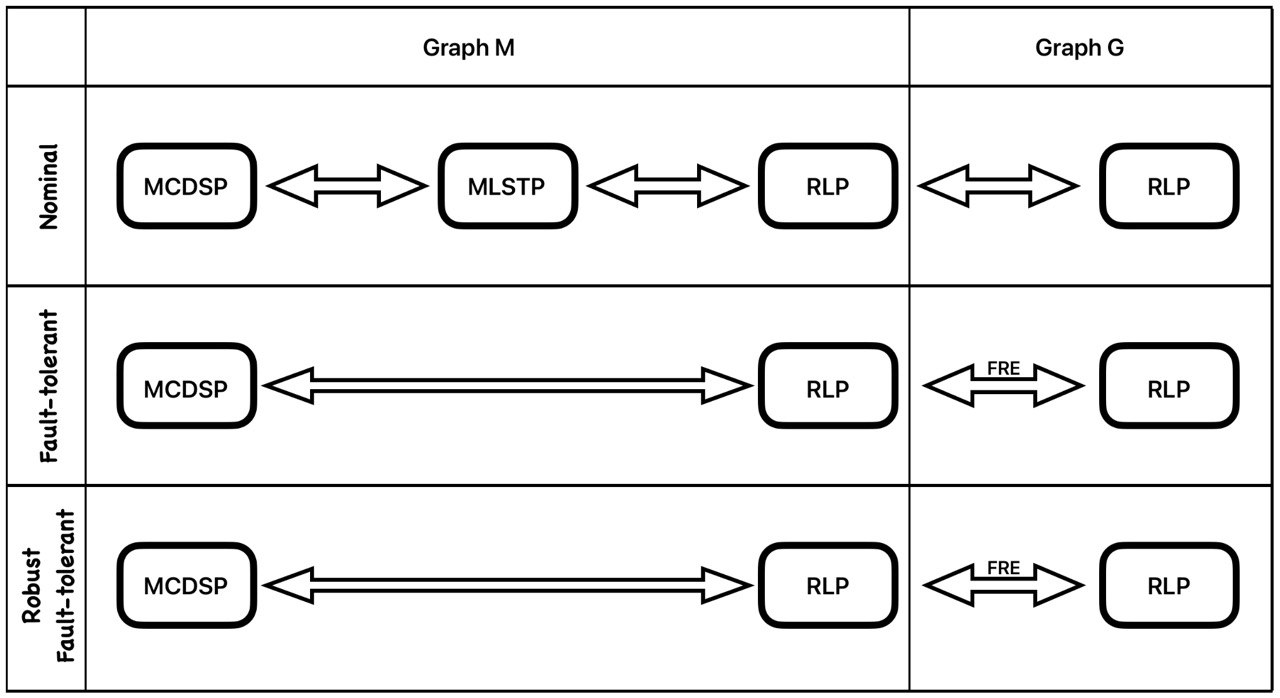}
		\end{center}
		\caption{Problem equivalency on $M$}\label{fig:flowchart}
	\end{figure}
	\section{Mathematical Models of RFTRLP}
	\label{sec:mathematical-models}
	In this section, we present two mathematical formulations for solving the robust FTMCDSP (RFTMCDSP), and, subsequently, the RFTRLP. 
	
	In the first formulation, we use the network flow problem to develop an integer programming (IP), which we call the \textit{flow-based} model. By representing the problem as a network flow, we leverage existing efficient algorithms and optimization techniques to find solutions that meet the desired criteria.
	
	In the second formulation, we integrate the minimum cut problem with Menger's theorem, and we call it the \textit{cut-based} model. This formulation focuses on identifying critical cuts in the network whose removal disrupts the connectivity required for the solution. By applying Menger's theorem, we can systematically determine these cuts and thus develop a robust solution framework.
	
	Both models involves a large number of variables and constraints; therefore, in the last part of this section, we develop an iterative solution approach only using the flow-based model to effectively solve large instances.
	
	Before proceeding to formulation, we describe a preprocessing step to practically reduce the computational time. 
	%
	\subsection{Preprocessing}
	\label{subsec:preprocessing}
	To tackle the problem effectively, we begin by applying a heuristic method to reduce the computational time needed. This reduction is crucial to simplify subsequent steps of the solution process. A key observation is that any node within the graph $M$ with a degree of two requires the placement of regenerators on both neighboring nodes because such a node must rely on its two neighbors for communication with the rest of the network. Consequently, the first step involves identifying all nodes in the network with a degree of two. This can be calculated in a polynomial time with $O$($n^2$). Once these nodes are identified, their neighboring nodes are added to the set $L$, which represents the initial placement of regenerators.
	
	Following this initial placement, we can start solving our models with a given partial solution referred to as the warm-start. This technique will reduce the solution time by reducing the branch-and-bound tree size for the out-of-the-box solvers like CPLEX and Gurobi.
	\subsection{Flow-Based Model}
	\label{subsec:flow-model}
	In the first solution approach, we present an IP formulation using the network flow problem. To write the desired model, we assume sending a unit flow between any two arbitrary nodes on $G$, through two edge-disjoint paths. Therefore, we can formulate the RFTMCDSP to solve the corresponding RFTRLP. In the following model, we use the decision variable $\pmb{x}\in\{0,1\}^n$ that takes the value of one if the corresponding node is included in the two-connected dominating set $L$ and zero otherwise.
	
	To write the constraints of the model, first, we need to ensure that there are two edge-disjoint paths between any pair of nodes included in $L$. As a result, we send a unit flow from any $p\in V$ to any $q\in V$ with $p\neq q$, through two edge-disjoint paths $f$ and $g$. In this sense, sum of the outgoing flows from $p$ is one unit more than the sum of incoming flows to $p$, if both nodes $p$ and $q$ are chosen to be included in $L$, otherwise the incoming flow to $p$ is equal to the outgoing flow of $p$. Similarly, sum of the incoming flows to $q$ is one unit more than the sum of outgoing flows from $q$, again if both nodes $p$ and $q$ are included in $L$, otherwise the outgoing flow of $q$ is equal to the incoming flow to $q$. In addition, the outgoing flow of any internal nodes of the $p-q$ path (if there is any internal node) is equal to its incoming flow. Thus we have
	\vspace{-0.2cm}
	
	\footnotesize
	\begin{align*}
		& \sum_{(p,j)\in E} f^{pj}_{pq} - \sum_{(j,p)\in E} f^{jp}_{pq} = x_p x_q & \forall p,q\in V \\
		& \sum_{(q,j)\in E} f^{qj}_{pq} - \sum_{(j,q)\in E} f^{jq}_{pq} = -x_p x_q & \forall p,q\in V \\
		& \sum_{(i,j)\in E} f^{ij}_{pq} - \sum_{(j,i)\in E} f^{ji}_{pq} = 0 & \forall p,q\in V,\; \forall i\in V\backslash(p,q)
	\end{align*}
	\normalsize
	
	We also have the exact same constraints for path $g$. Obviously, we have a non-linearity in our constraint, since the difference between the incoming and outgoing flow of a node depends on the existence of $p$ and $q$ in $L$. Consequently, we introduce a variable using the McCormick law to remove the non-linearity as follows
	\vspace{-0.2cm}
	
	\footnotesize
	\begin{align*}
		& t_{pq} = x_p x_q & \forall p,q\in V\\
		& t_{pq} \le \min(x_p,x_q) & \forall p,q\in V\\
		& t_{pq} \ge \max(0,x_p - (1-x_q))& \forall p,q\in V
	\end{align*}
	\normalsize
	this way, $t_{pq}=1$ if and only if both $x_p=1$ and $x_q=1$.
	
	Then, it is important to ensure any node is connected to at least two nodes of $L$, which can be formulated in the form of the following constraint.
	\vspace{-0.2cm}

	\footnotesize
	\begin{align*}
		& \quad \textstyle\sum_{j\in\mathcal{N}_i} x_j \ge 2, && \forall i \in V
	\end{align*}
	\normalsize
	where $\mathcal{N}_i$ represents the neighboring nodes of $i$ on the transformed graph $M$. Moreover, it is necessary that two paths $f$ and $g$ are edge-disjoint, thus we have:
	\vspace{-0.2cm}
	
	\footnotesize
	\begin{align*}
		& \quad f^{ij}_{pq} + g^{ij}_{pq} \le 1, && \forall p,q\in V,\; \forall (i,j)\in E
	\end{align*}
	\normalsize
	
	Finally, we add the next constraint to satisfy the FRE criterion
	\vspace{-0.2cm}
	
	\footnotesize
	\begin{align*}
		& \quad \textstyle\sum_{j\in \mathcal{N}_i} x_j \geq \sum_{k\in\mathcal{N}_i^{\prime}} x_k +1 & \forall i\in \overline{V}
	\end{align*}
	\normalsize
	where $\mathcal{N}_i^{\prime}$ consists of nodes lying on all paths added to $M$ using the shortest path algorithm and start with $i$. Now we can formulate the flow-based model, which we refer to it as IP-FB, as follows
	\vspace{-0.2cm}
	
	\footnotesize
	\begin{align*}
		\min & \;\; z \displaybreak[1]\\
		\text{s.t.} & \;\; \textstyle\sum_{i\in V} c^k_i x_i \le z, &\forall k \in [N] \displaybreak[1]\\
		& \;\; \textstyle\sum_{(p,j)\in E} f^{pj}_{pq} - \sum_{(j,p)\in E} f^{jp}_{pq} = t_{pq} & \forall p,q\in V \displaybreak[1]\\
		& \;\; \textstyle\sum_{(q,j)\in E} f^{qj}_{pq} - \sum_{(j,q)\in E} f^{jq}_{pq} = -t_{pq} & \forall p,q\in V \displaybreak[1]\\
		& \;\; \textstyle\sum_{(i,j)\in E} f^{ij}_{pq} - \sum_{(j,i)\in E} f^{ji}_{pq} = 0 & \forall p,q\in V,\; \forall i\in V\backslash(p,q)\\
		& \;\; \textstyle\sum_{(p,j)\in E} g^{pj}_{pq} - \sum_{(j,p)\in E} g^{jp}_{pq} = t_{pq} & \forall p,q\in V \displaybreak[1]\\
		& \;\; \textstyle\sum_{(q,j)\in E} g^{qj}_{pq} - \sum_{(j,q)\in E} g^{jq}_{pq} = -t_{pq} & \forall p,q\in V \displaybreak[1]\\
		& \;\; \textstyle\sum_{(i,j)\in E} g^{ij}_{pq} - \sum_{(j,i)\in E} g^{ji}_{pq} = 0 & \forall p,q\in V,\; \forall i\in V\backslash(p,q)\\
		& \;\; \textstyle\sum_{j\in\mathcal{N}_i} x_j \ge 2, & \forall i \in V \displaybreak[1]\\
		& \;\; \textstyle\sum_{j\in \mathcal{N}_i} x_j \geq \sum_{k\in\mathcal{N}_i^{\prime}} x_k +1, & \forall i\in \overline{V} \displaybreak[1]\\
		& \;\; t_{pq} \le \min(x_p,x_q), & \forall p,q\in V  \displaybreak[1]\\
		& \;\; t_{pq} \ge \max(0,x_p - (1-x_q)), & \forall p,q\in V  \displaybreak[1]\\
		& \;\; f^{ij}_{pq} + g^{ij}_{pq} \le 1, & \forall p,q\in V,\; \forall (i,j)\in E  \displaybreak[1]\\
		& \;\; f^{ij}_{pq} , g^{ij}_{pq} \le x_i, & \forall p,q\in V,\; \forall (i,j)\in E  \displaybreak[1]\\
		& \;\; f^{ij}_{pq} , g^{ij}_{pq} \le x_j, & \forall p,q\in V,\; \forall (i,j)\in E  \displaybreak[1]\\
		& \;\; f^{ij}_{pq} , g^{ij}_{pq} \le x_p, & \forall p,q\in V,\; \forall (i,j)\in E  \displaybreak[1]\\
		& \;\; f^{ij}_{pq} , g^{ij}_{pq} \le x_q, & \forall p,q\in V,\; \forall (i,j)\in E  \displaybreak[1]\\
		& \;\; f^{ij}_{pq} \in [0,1], & \forall p,q\in V,\; \forall (i,j)\in E \displaybreak[1]\\
		& \;\; g^{ij}_{pq} \in [0,1], & \forall p,q\in V,\; \forall (i,j)\in E \displaybreak[1]\\
		& \;\; x_i \in \{0,1\}, & \forall i \in V \displaybreak[1]\\
		& \;\; t_{pq} \in \{0,1\}, & \forall p,q\in V 
	\end{align*}
	\normalsize
	\subsection{Cut-Based Model}\label{subsec:cut-model}
	
	In this section, we present the second model whose solution returns the RFTMCDSP on graph $G$, and thus solves the RFTRLP. This formulation is obtained having considered the minimum cut problem combined with the Menger's theorem \cite{menger1927allgemeinen}. Based on the Menger's theorem, the number of edge-disjoint paths between two nodes $p$ and $q$ on a graph is equal to the minimum number of cutting edges that separate these two vertices, which we refer to it as $p-q$ cut.
	
	It is necessary to remind that the cut $[S,V\backslash S]$, separating the vertices $p,q\in V$, partitions the set of vertices $V$ into two non-empty subsets $S$ and $V\backslash S$ such that $p\in S$ and $q\in V\backslash S$. For simplicity, we use $S$-cut to refer to the cut $[S,V\backslash S]$.
	
	Analogously to the model introduced in Section~{\ref{subsec:flow-model}}, we use the decision variable $\pmb{x}\in\{0,1\}^n$ corresponding to the nodes of the network which is equal to one if the respective node is included in $L$ and zero otherwise. To formulate this problem we first formulate constraints related to the Menger's theorem. Therefore, we introduce the following variables
	\vspace{-0.2cm}
	
	\footnotesize
	\begin{align*}
		& p_s = \min \Big(1, \sum_{i\in S} x_i \Big) & \forall S\subseteq V \\
		& q_s = \min \Big(1, \sum_{i\in V\backslash S} x_i \Big) & \forall S\subseteq V
	\end{align*}
	\normalsize
	where both $p_s, q_s\in\{0,1\}$ for all $S\subseteq V$. Here, $p_s=1$ if at least one node of $S$ is included in the solution subset and $p_s=0$ otherwise. In a similar way, $q_s=1$ if at least one node of $V\backslash S$ is included in the solution subset and $p_s=0$ otherwise.
	
	According to the assumption, to ensure that there are two edge-disjoint paths between all pair of nodes included in the selected subset of nodes, it is sufficient that the size of the cut for any two vertices of the minimum two-connected dominating set is at least equal to two. We can formulate this constraint as follows
	\vspace{-0.2cm}
	
	\footnotesize
	\begin{align*}
		& \sum_{(i,j)\in E,\; i\in S,\; j \in V\backslash S} x_i x_j \ge 2 p_s q_s & \forall S\subseteq V
	\end{align*}
	\normalsize
	
	In this constraint, we have non-linearity on both the left-hand-side (lhs) and the right-hand-side (rhs). Both lhs and rhs of this constraint can be linearized using the McCormick law. For the lhs we have 
	\vspace{-0.2cm}
	
	\footnotesize
	\begin{align*}
		& r_{ij} = x_i x_j & \forall S\subseteq V, \; \forall i\in S, \; \forall j\in V\backslash S \\
		& r_{ij} \le \min \big(x_i,x_j\big) & \forall S\subseteq V, \; \forall i\in S, \; \forall j\in V\backslash S \\
		& r_{ij} \ge \max \big(0,x_i-(1-x_j)\big) & \forall S\subseteq V, \; \forall i\in S, \; \forall j\in V\backslash S
	\end{align*}
	\normalsize
	the rhs also can be linearized in the following way
	
	\footnotesize
	\begin{align*}
		& t_s = p_s q_s & \forall S\subseteq V \\
		& t_s \le \min \big(p_s , n q_s\big) & \forall S\subseteq V \\
		& t_s \ge \max \big( 0, p_s - n(1- q_s) \big) & \forall S\subseteq V
	\end{align*}
	\normalsize
	where $\lvert V\rvert =n$. Furthermore, the constraints to satisfy FRE criterion as well as the two-connectivity and edge disjoint requirements are formulated similar to their corresponding constraints in Section~\ref{subsec:flow-model}. Hence, the cut-based IP formulation, which we refer to it as IP-CB, could be written in the following way
	\vspace{-0.2cm}
	
	\footnotesize
	\begin{align*}
		\min \quad & z \displaybreak[1]\\
		\text{s.t.} \quad & \sum_{i\in V} c^k_i x_i \le z & \forall k\in[N] \displaybreak[1]\\
		& p_s = \min \Big(1, \sum_{i\in S} x_i \Big) & \forall S\subseteq V \displaybreak[1]\\
		& q_s = \min \Big(1, \sum_{i\in V\backslash S} x_i \Big) & \forall S\subseteq V \displaybreak[1]\\
		& t_s \le \min \big(p_s , n q_s\big) & \forall S\subseteq V \displaybreak[1]\\
		& t_s \ge \max \big( 0, p_s - n(1-q_s) \big) & \forall S\subseteq V \displaybreak[1]\\
		& r_{ij} \le \min \big(x_i,x_j\big) & \forall S\subseteq V, \; \forall i\in S, \; \forall j\in V\backslash S \displaybreak[1]\\
		& r_{ij} \ge \max \big(0,x_i-(1-x_j)\big) & \forall S\subseteq V, \; \forall i\in S, \; \forall j\in V\backslash S \displaybreak[1]\\
		& \sum_{(i,j)\in E,\; i\in S,\; j \in V\backslash S} r_{ij} \ge 2 t_s & \forall S\subseteq V \displaybreak[1]\\
		& \sum_{j\in\mathcal{N}_i} x_j \ge 2 & \forall i \in V \displaybreak[1]\\
		& \sum_{j\in \mathcal{N}_i} x_j \geq \sum_{k\in\mathcal{N}_i^{\prime}} x_k + 1 & \forall i\in \overline{V} \displaybreak[1]\\
		& t_s, \; p_s, \; q_s \in \{0,1\} & \forall S\subseteq V \displaybreak[1]\\
		& r_{ij} \in \{0,1\} & \forall i, j \in V \displaybreak[1]\\
		& x_i \in \{0,1\} & \forall i \in V 
	\end{align*}
	\normalsize
	
	Although the IP-CB formulation contains an exponential number of cut constraints indexed by all subsets $S \subseteq V$, it is solved efficiently using a branch-and-cut framework. In particular, the model is implemented in CPLEX, which employs branch-and-cut by default, and the connectivity constraints are handled implicitly through dynamic cut generation. Rather than enumerating all cut constraints, only violated constraints are identified via minimum-cut separation and added on demand during the solution process.
	
	\subsection{Iterative Method}\label{subsec:iterative-approach}
	
	As IP-CB has exponential number of variables and constraint, it cannot be applied for large-scale instances, for it needs a very large memory size as well as solution time to be solved. In addition, IP-FB also generates finitely many variables and constraints and thus it faces the same problem to be solved optimally for instances with many nodes. Therefore, in order to practically solve large-scale instances we propose an iterative approach to solve the RFTRLP optimally. This model is divided into a master problem (MP) which can be written as follows:
	\vspace{-0.2cm}
	
	\footnotesize
	\begin{align*}
		\min \quad & z \displaybreak[1]\\
		\text{s.t.} \quad & \sum_{i\in V} c^k_i x_i \le z & \forall k\in[N] \displaybreak[1]\\
		& \sum_{j\in \mathcal{N}_i} x_j \ge 2 & \forall i \in V \displaybreak[1]\\
		& \sum_{j\in \mathcal{N}_i} x_j \geq \sum_{k\in\mathcal{N}_i^{\prime}} x_k +1 & \forall i\in \overline{V} \displaybreak[1]\\
		& x_i \in \{0,1\} & \forall i \in V \displaybreak[1]
	\end{align*}
	\normalsize
	where $\mathcal{N}_i$ and $\mathcal{N}_i^{\prime}$ are defined in the same way as in Sections~\ref{subsec:flow-model} and \ref{subsec:cut-model}. Note that the solution obtained by solving the master problem might not be an optimal or even a feasible solution to the RFTRLP. See Example~\ref{exp:mp-solution}.
	
	\begin{example}\label{exp:mp-solution}
		In Figure~\ref{fig:mp-vs-2rlp} the cost of installing regenerators on nodes 3 and 7 is ten while the cost of deploying a regenerator in the remaining nodes are one. The optimal solution of the master problem is $\{1,2,4,5,6,8\}$ with cost 6, while this is not a feasible solution to the RFTRLP, for it is not possible to communicate from any node of $\{1,2,8\}$ to any node of $\{4,5,6\}$.
		\begin{figure}[htbp]
			\begin{center}
				\includegraphics[width=0.4\textwidth]{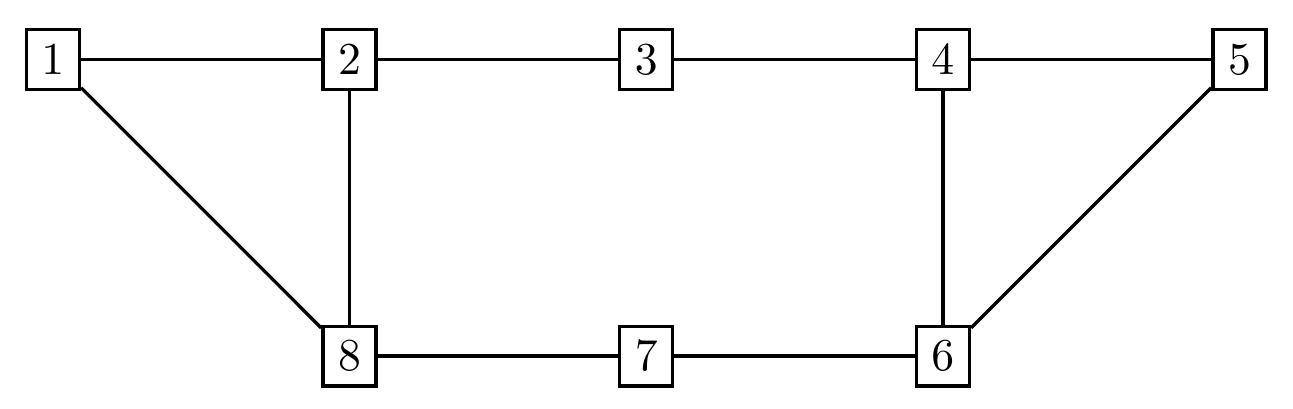}
			\end{center}
			\caption{Ineqivalency of solution to MP and RFTRLP}\label{fig:mp-vs-2rlp}
		\end{figure}
	\end{example}
	
	The solution found using MP might not be connected, therefore we need to assure that the two-edge connectivity is guaranteed through considering a subproblem (SP). We can introduce the SP using either of IP-FB or IP-CB. However, as the IP-CB grows exponentially, we only focus on the IP-FB for the SP. Therefore, we first fix the solution of MP in SP and solve it to find the nodes that guarantee the two-edge connectivity in the union of the found solution by both SP and MP. Hence, we can formulate the SP as follows:
	\vspace{-0.2cm}
	
	\footnotesize
	\begin{align*}
		\min \quad & z \displaybreak[1]\\
		\text{s.t.} \quad & \sum_{i\in V} c^k_i x_i \le z & \forall k \in [N] \displaybreak[1]\\
		& \sum_{(p,j)\in E} f^{pj}_{pq} - \sum_{(j,p)\in E} f^{jp}_{pq} = t_{pq} & \forall p,q\in V \displaybreak[1]\\ 
		& \sum_{(q,j)\in E} f^{qj}_{pq} - \sum_{(j,q)\in E} f^{jq}_{pq} = -t_{pq} & \forall p,q\in V \displaybreak[1]\\
		& \sum_{(i,j)\in E} f^{ij}_{pq} - \sum_{(j,i)\in E} f^{ji}_{pq} = 0 & \forall p,q\in V,\; \forall i\in V\backslash(p,q) \displaybreak[1]\\
		& \sum_{(p,j)\in E} g^{pj}_{pq} - \sum_{(j,p)\in E} g^{jp}_{pq} = t_{pq} & \forall p,q\in V \displaybreak[1]\\ 
		& \sum_{(q,j)\in E} g^{qj}_{pq} - \sum_{(j,q)\in E} g^{jq}_{pq} = -t_{pq} & \forall p,q\in V \displaybreak[1]\\
		& \sum_{(i,j)\in E} g^{ij}_{pq} - \sum_{(j,i)\in E} g^{ji}_{pq} = 0 & \forall p,q\in V,\; \forall i\in V\backslash(p,q) \displaybreak[1]\\
		& t_{pq} \le \min(x_p,x_q) & \forall p,q\in V \displaybreak[1]\\
		& t_{pq} \ge \max(0,x_p - (1-x_q)) & \forall p,q\in V \displaybreak[1]\\
		& f^{ij}_{pq} + g^{ij}_{pq} \le 1 & \forall p,q\in V,\; \forall (i,j)\in E \displaybreak[1]\\
		& f^{ij}_{pq} \; , \; g^{ij}_{pq} \le x_i & \forall p,q\in V,\; \forall (i,j)\in E \displaybreak[1]\\
		& f^{ij}_{pq} \; , \; g^{ij}_{pq} \le x_j & \forall p,q\in V,\; \forall (i,j)\in E \displaybreak[1]\\
		& f^{ij}_{pq} \; , \; g^{ij}_{pq} \le x_p & \forall p,q\in V,\; \forall (i,j)\in E \displaybreak[1]\\
		& f^{ij}_{pq} \; , \; g^{ij}_{pq} \le x_q & \forall p,q\in V,\; \forall (i,j)\in E \displaybreak[1]\\
		& f^{ij}_{pq} \in [0,1] & \forall p,q\in V,\; \forall (i,j)\in E \displaybreak[1]\\
		& g^{ij}_{pq} \in [0,1] & \forall p,q\in V,\; \forall (i,j)\in E \displaybreak[1]\\
		& x_i \in \{0,1\} & \forall i \in V \displaybreak[1]\\
		& t_{pq} \in\{0,1\} & \forall p,q\in V 
	\end{align*}
	\normalsize
	
	In addition, the solution obtained using SP (combination of solutions found in MP and SP) is a feasible solution, but not optimal. Therefore we fix the solution SP$\backslash$MP into MP and continue this procedure until the solution obtained by SP is the same as the solution found by its previous SP. In other words, the solution of SP is also two-edge connected. We refer to this iterative approach as IT-FB.
	
	\begin{theorem}\label{theorem:IT-FB-convergence}
		The IT-FB method converges to the optimal solution within finite iterations.
	\end{theorem}
	\begin{proof}
		Let $L^t$ and $K^t$ denote the solutions obtained at iteration $t\in\mathcal{T}$ by solving MP and SP, respectively. By construction, the set $L^t$ is a minimum-cost solution that ensures, under the FRE constraint, every node in the network is adjacent to at least two nodes equipped with regenerators.
		
		After iteration $t$, the set $K^t$ is fixed and MP is solved again to obtain $L^{t+1}$. Since $K^t$ guarantees that the solution at iteration $t$ is two-edge-connected, it follows that $|K^t|\geq 2$. At the beginning of iteration $t+1$, if no node in the network is adjacent to at least two nodes in $K^t$, then the solution of MP remains unchanged, i.e., $L^{t+1}=L^t$. In this case, the union $K^t \cup L^t$ constitutes an optimal solution, and the algorithm terminates.
		
		Otherwise, if there exists at least one node that is adjacent to two nodes in $K^t$, then the MP yields a strictly smaller solution set, i.e., $L^{t+1}\subset L^t$, and the algorithm proceeds to the next iteration. This process is repeated until an iteration $t'>t$ is reached such that $L^{t'+1}=L^{t'}$.
		
		Since the number of nodes in the network is finite and the sequence $\{L^t\}$ is monotonically non-increasing, the algorithm must terminate after a finite number of iterations.
	\end{proof}
	
	\section{Experimental Results}\label{sec:experiments}
	
	In this section, we check the performance of our solution approaches. To the best of our knowledge, there exists no similar work to compare our methods with; therefore, we use the most theoretically relevant work in \cite{li2020branch} and adjust its IP accordingly (we refer to it as IP-LA) to make a valid comparison. The best performance of solution methods for large-scale instances introduced in \cite{li2020branch} belongs to the branch-and-Benders decomposition using an incumbent start. However, this cannot be compared to our iterative approach since their methods cannot be adjusted for a discrete uncertainty set affecting its objective function.
	
	We note that the exact formulations, IP-FB and IP-CB, are primarily introduced to provide rigorous mathematical models, and serve as optimal benchmarks for validation purposes. Due to the presence of complex connectivity constraints and, in the case of IP-CB, an exponential number of potential cuts, these formulations are inherently limited to small-sized instances. Accordingly, their evaluation is restricted to instances where optimality can be reached (in an acceptable time) and meaningful comparisons among formulations can be performed. The scalability to larger networks is addressed through the proposed iterative solution method, which is designed to overcome the computational limitations of the exact models and is evaluated separately on substantially larger instances.
	
	In this paper, we use two different instance generation methods, called Gen-1 and Gen-2. In Gen-1, the lengths of the edges are chosen from $\{100,101,\ldots,300\}$, and the cost of installing a regenerator on a node in each scenario (the uncertainty set $\U$) is selected from $\{100,101,\ldots,200\}$. In Gen-2, we use the idea of the sampling method introduced in \cite{goerigk2024benchmarking} to generate hard instances. Thus, with a ten-percent probability, edge lengths are chosen from $\{100,101,\ldots,120\}$ and with ninety-percent probability from $\{180,181,\ldots,200\}$. Moreover, the cost of each node in the uncertainty set is chosen from $\{180,181,\ldots,200\}$, i.e., the price of each node is selected from a tighter range. In both methods, all parameters are generated in a randomly uniform way, and we set $d_{max}=300$. Moreover, we introduce a new parameter called Dens, representing the generated graph's density. Therefore, the value of Dens ranges between 0 and 1 and illustrates the ratio of the number of edges in the generated graph to the number of edges in the corresponding complete graph.
	
	To avoid long process times, we generate large instances in a way that the all-pairs shortest path cannot add any new edges. Thus, for both Gen-1 and Gen-2, the edge lengths are derived from $\{151, 152,\ldots,300\}$ when all other parameter selections remain the same.
	
	We divide our experiments into two sets, with small and large instances. In the experiments on small instances, we compare IP-FB, IP-CB, and IP-LA. Specifically, we first solve instances with a fixed number of scenarios and a varying number of nodes (Exp-1), and then solve instances where the number of nodes is fixed and the number of scenarios varies (Exp-2). For Exp-1, we set $N=20$ and $n=\{10,11,\ldots,18\}$, while for Exp-2 we solve instances with $n=12$ and $N=\{20,30,\ldots,100\}$. In both Exp-1 and Exp-2, we fix Dens to 0.6. In addition, we conduct a third experiment (Exp-3) to compare the strength of the formulations in terms of their linear relaxations. Specifically, we evaluate the convex hull approximations by solving the LP relaxations of the proposed IP formulations. In this experiment, we fix $N=10$, consider $n=\{8,9,\ldots,16\}$, and set Dens to 0.4. In these experiments we solve 100 instances using a 1800-second time limit for Exp-1, and 600-second time limit for Exp-2 and Exp-3.
	
	In addition, for the large-scale experiments, we only evaluate the IT-FB as the most efficient solution method. Similar to the small experiments, we first solve instances with a fixed number of scenarios (Exp-4) with $N=20$ and a varying number of nodes $n=\{50,60,\ldots,150\}$. Then, we fix the number of nodes (Exp-5) to $n=100$ and vary the number of scenarios to $N=\{10,20,\ldots,100\}$. In both Exp-4 and Exp-5, we solve 100 instances considering three values of 0.4, 0.6, and 0.8 for graph density, as well as a 600-second time limit for Exp-4 and a 1200-second time limit for Exp-5.
	
	To evaluate solution times for each problem, and in addition to presenting the normal solution times, we use performance profiles as introduced in \cite{dolan2002benchmarking}. We briefly recall this concept: Let $\cS$ be the set of considered models, $\K$ the set of instances and $t_{k,s}$ the run time of model $s$ on instance $k$. We assume $t_{k,s}$ is set to infinity (or large enough) if model $s$ does not solve instance $k$ within the time limit. The percentage of instances for which the performance ratio of solver $s$ is within a factor $\tau \geq 1$ of the best ratio of all solvers is given by:
	\vspace{-0.2cm}
	
	\footnotesize
	\[k_s (\tau) = \frac{1}{\lvert \K \rvert} \; \Bigg\lvert \Bigg\{ k\in \K \; \vert \; \frac{t_{k,s}}{\text{min}_{\hat{s}\in \cS} t_{k,\hat{s}}}  \leq \tau \Bigg\} \Bigg \lvert  \]
	\normalsize
	Hence, the function $k_s$ can be viewed as the distribution function for the performance ratio, which is plotted in a performance profile for each model. Note that larger values of $k_s$ indicate that the corresponding solution method dominates the others on a greater proportion of instances. Therefore, the higher the value of $k_s$, the better the method is.
	
	For all experiments, we use CPLEX version 22.11 on a GenuineIntel pc-i440fx-7.2 CPU computer server running at 2.00 GHz with 15 GB of RAM. Additionally, only one core is used for each instance. Code was implemented in C++.
	
	\subsection{Results of Small Instances}
	
	
	Tables~\ref{table:exp1-gen1} and \ref{table:exp1-gen2} summarize the average costs and solution times of Exp-1. In these tables, the column \textit{SHP} indicates the time required to solve the shortest-path IP on the original graph in order to obtain the transformed graph $M$. Note that for IP-LA, the shortest-path IP must be solved separately for each fault scenario introduced in \cite{li2020branch}. Furthermore, the column \textit{IP} reports the average time required to solve the corresponding IP formulation on the transformed graph $M$. The results clearly show that the solution times of both SHP and IP increase as the number of nodes increases for all IP formulations.
	
	It can be seen that instances generated using Gen-2 are slightly harder to solve. In addition, although IP-CB is the fastest formulation for the smallest instances, IP-FB dominates the other two approaches, both in terms of SHP time and, more importantly, in terms of IP solution times. Both IP-FB and IP-CB show significantly lower IP times compared to IP-LA.
	\vspace{-0.2cm}
	\begin{table}[htb]
		\begin{center}
			\caption{Average results of Exp-1 for Gen-1 instances.}\label{table:exp1-gen1}
			\begin{tabular}{|c|c|c|c|c|c|c|c|c|c|c|}
				\hline
				& & \multicolumn{2}{|c|}{IP-FB}&\multicolumn{2}{|c|}{IP-CB}&\multicolumn{2}{|c|}{IP-LA} \\
				\cline{3-8}
				Nodes & Cost & SHP & IP & SHP & IP & SHP & IP \\
				\hline
				10 & 545.59 & 0.05 & 0.51 & 0.06 & 0.26  & 1.67 & 1.861 \\
				11 & 540.04 & 0.06 & 0.87 & 0.07 & 0.61  & 2.27 & 4.044 \\
				12 & 538.90 & 0.09 & 1.42 & 0.09 & 1.49  & 3.52 & 11.26 \\
				13 & 540.15 & 0.12 & 2.39 & 0.12 & 3.55  & 5.37 & 25.97 \\
				14 & 534.47 & 0.15 & 3.90 & 0.15 & 8.99  & 8.01 & 62.13 \\
				15 & 533.20 & 0.18 & 6.20 & 0.18 & 18.5  & 11.9 & 145.9 \\
				16 & 531.80 & 0.24 & 9.80 & 0.24 & 38.2  & 17.6 & 331.2 \\
				17 & 530.90 & 0.32 & 15.5 & 0.32 & 78.6  & 25.4 & 721.3 \\
				18 & 529.50 & 0.38 & 24.5 & 0.38 & 155 & 36.8 & 1502 \\
				\hline
			\end{tabular}
		\end{center}
	\end{table}
	
	Figures~\ref{fig:exp1-pp-gen1} and \ref{fig:exp1-pp-gen2} illustrate the performance profiles of the overall process times (computed as the sum of SHP and IP times), as well as the ratios of IP-LA times to those of our proposed methods. Note that \textit{SHP-LA} represents the ratio of the SHP time of IP-LA to that of our methods. Similarly, \textit{IP-LA/FB} and \textit{IP-LA/CB} denote the ratios of the IP time of IP-LA to the IP times of the FB and CB formulations, respectively.
	\vspace{-0.2cm}
	\begin{figure}[htb]
		\begin{center}
			\subfigure[Process Time]{\includegraphics[width=0.23\textwidth]{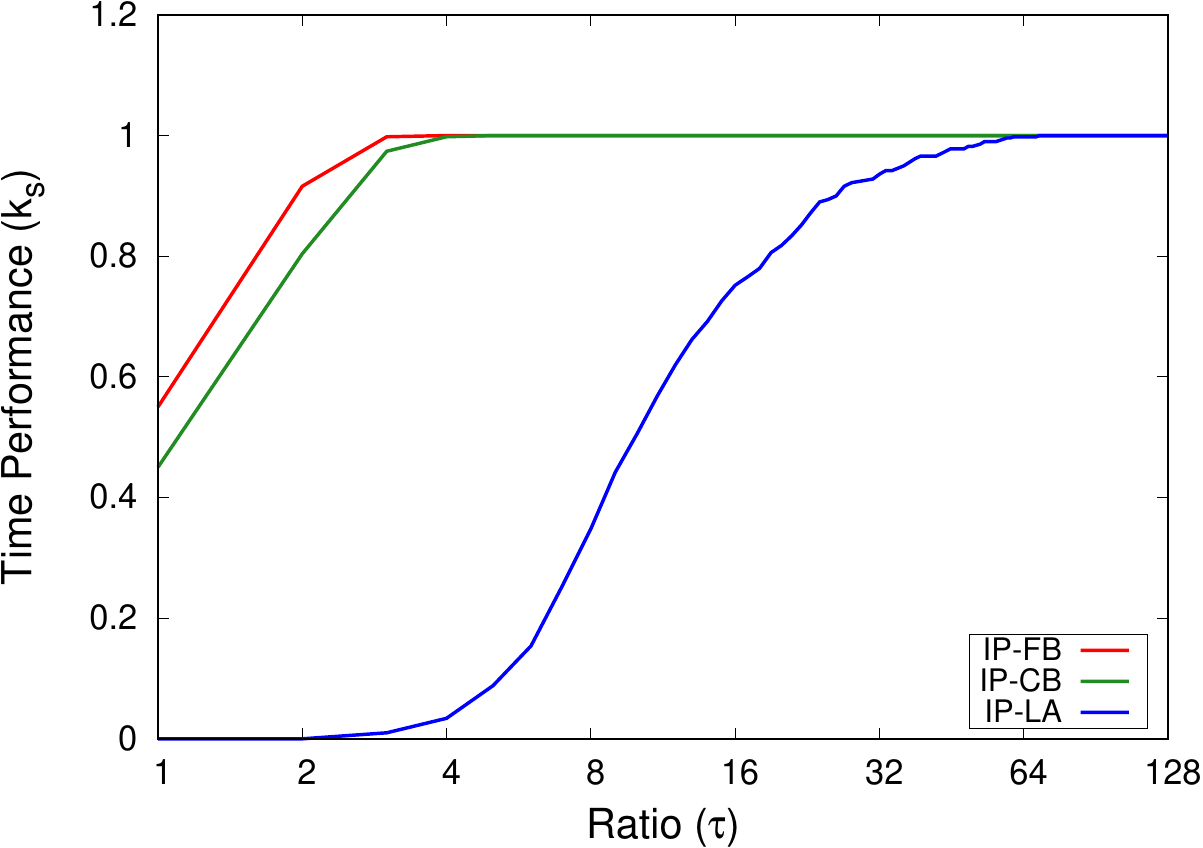}}
			\subfigure[Time Ratio]{\includegraphics[width=0.23\textwidth]{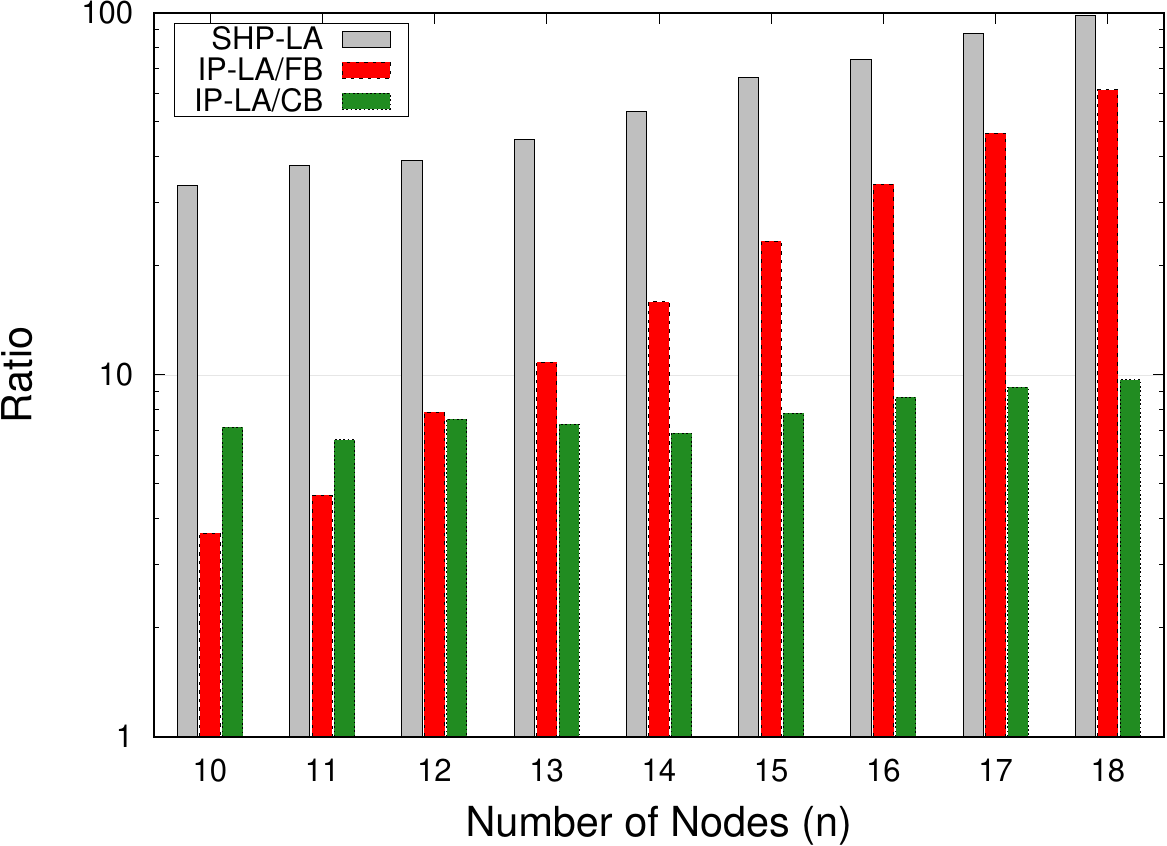}}
		\end{center}
		\caption{Performance profile of process time for IPs (a), and ratio of SHP and IP times of IP-LA compared to our IP models (b) in Exp-1 for Gen-1.}\label{fig:exp1-pp-gen1}
	\end{figure}

	The corresponding performance profiles demonstrate that IP-FB and IP-CB not only dominate IP-LA in terms of average IP and SHP times, but also achieve lower instance-wise process times, leading to higher $k_s$ values. Nevertheless, the relative performance of IP-FB and IP-CB is mixed. It should be noted that this happens only for small instances, due to the exponential number of variables and constraints in IP-CB formulation.
	\vspace{-0.2cm}
	\begin{table}[htb]
		\begin{center}
			\caption{Average results of Exp-1 for Gen-2 instances.}\label{table:exp1-gen2}
			\begin{tabular}{|c|c|c|c|c|c|c|c|}
				\hline
				& & \multicolumn{2}{|c|}{IP-FB}&\multicolumn{2}{|c|}{IP-CB}&\multicolumn{2}{|c|}{IP-LA} \\
				\cline{3-8}
				Nodes & Cost & SHP & IP & SHP & IP & SHP & IP \\
				\hline
				10 & 657.29 & 0.05 & 0.45 & 0.05 & 0.23  & 1.33 & 1.71 \\
				11 & 643.30 & 0.07 & 0.79 & 0.07 & 0.59  & 2.25 & 4.66 \\
				12 & 624.21 & 0.09 & 1.28 & 0.09 & 1.39  & 3.51 & 12.4 \\ 
				13 & 628.92 & 0.12 & 2.26 & 0.12 & 3.42  & 5.37 & 30.1 \\
				14 & 616.24 & 0.15 & 3.72 & 0.15 & 8.61  & 8.07 & 72.1 \\
				15 & 610.10 & 0.18 & 5.90 & 0.18 & 17.9  & 11.8 & 155.3 \\
				16 & 603.70 & 0.24 & 9.30 & 0.24 & 36.5  & 17.5 & 350.7 \\
				17 & 597.90 & 0.32 & 14.8 & 0.32 & 74.3  & 25.2 & 760.2 \\
				18 & 592.60 & 0.38 & 23.5 & 0.38 & 148   & 36.5 & 1550 \\
				\hline
			\end{tabular}
		\end{center}
	\end{table}
	\vspace{-0.2cm}
	
	The SHP and IP time ratios of the IP-LA formulation relative to our proposed methods indicate that, as the number of nodes increases, these ratios grow significantly for the SHP and IP-FB models (reaching values of up to approximately 100), while they remain nearly constant for IP-CB (reaching to almost 8). This clearly highlights the advantages of both proposed methods, and in particular IP-FB, for both Gen-1 and Gen-2.
	
	\begin{figure}[htb]
		\begin{center}
			\subfigure[Process Time]{\includegraphics[width=0.23\textwidth]{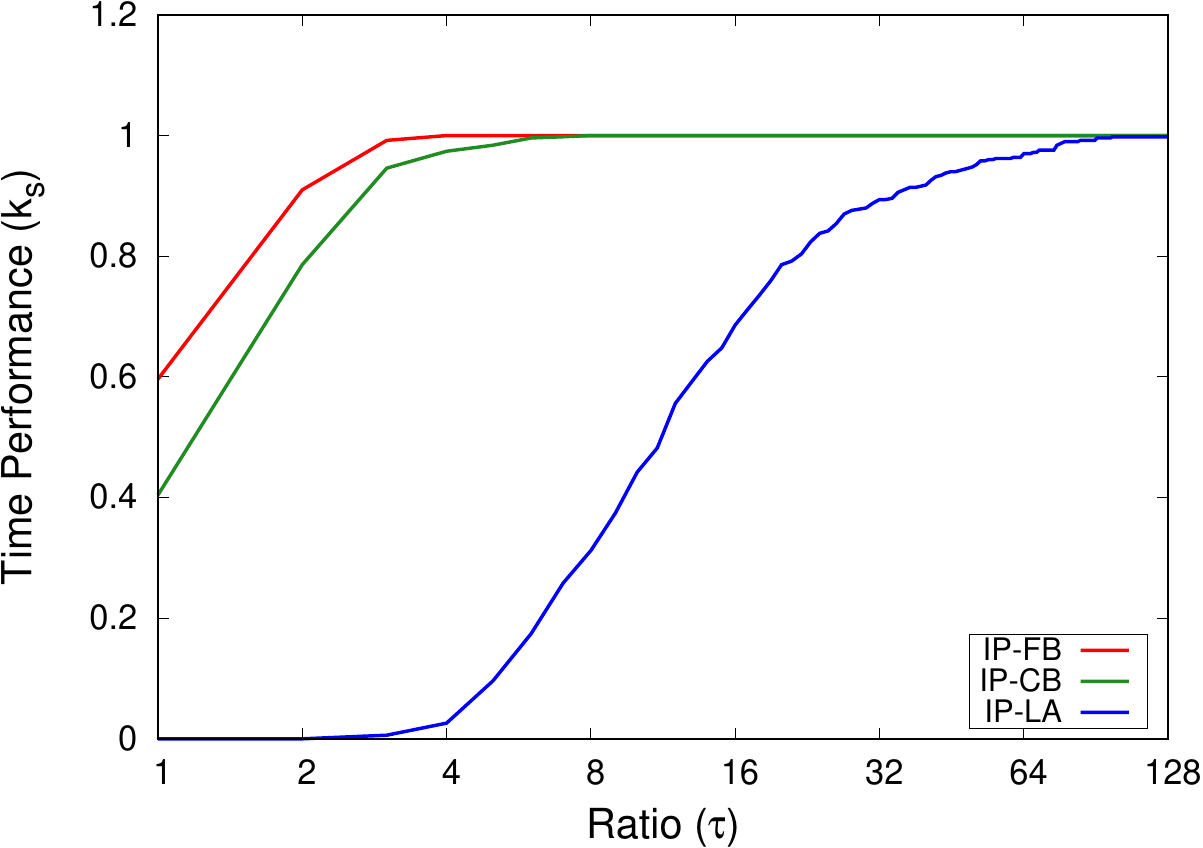}}
			\subfigure[Time Ratio]{\includegraphics[width=0.23\textwidth]{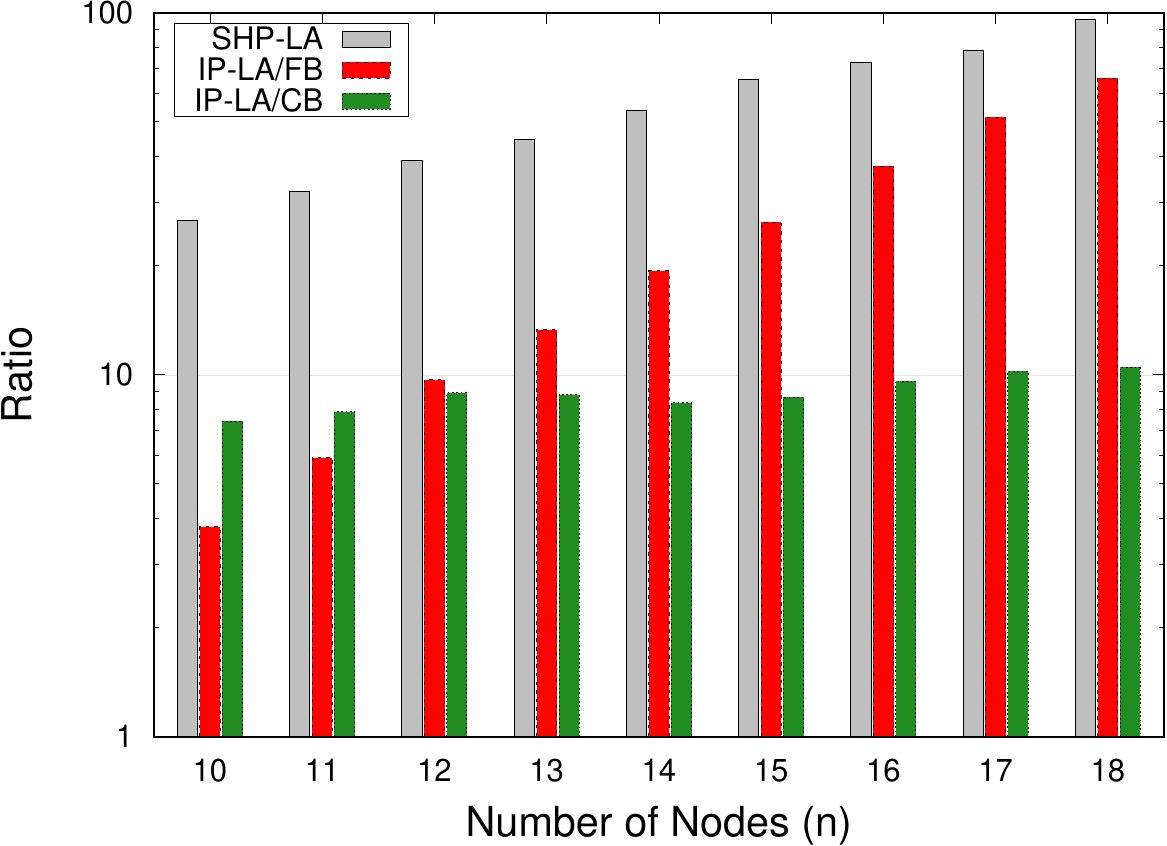}}
		\end{center}
		\caption{Performance profile of process time for solution methods (a), and relative SHP and IP times compared to IP-LA (b) in Exp-1 for Gen-2.}\label{fig:exp1-pp-gen2}
	\end{figure}
	
	
	The average results reported in Tables~\ref{table:exp2-gen1} and \ref{table:exp2-gen2} indicate that varying the number of scenarios has little impact on the IP solution times of all three formulations, as the considered instances are relatively small and can be solved very efficiently. Although IP-CB achieves the lowest average IP time for Gen-1, IP-FB exhibits the best overall performance for Gen-2. In terms of SHP times, our proposed methods are nearly 40 times faster than the SHP time required by IP-LA. Moreover, the proposed IP formulations are approximately 8 times faster in Gen-1 and up to 10 times faster in Gen-2 compared to IP-LA.
	
	\begin{table}[htb]
		\begin{center}
			\caption{Average results for Exp-2 under Gen-1 instance generation.}\label{table:exp2-gen1}
			\begin{tabular}{|c|c|c|c|c|c|c|c|}
				\hline
				& & \multicolumn{2}{|c|}{IP-FB} & \multicolumn{2}{|c|}{IP-CB} & \multicolumn{2}{|c|}{IP-LA} \\
				\cline{3-8}
				Scenarios & Cost & SHP & IP & SHP & IP & SHP & IP \\
				\hline
				20  & 538.90 & 0.09 & 1.41 & 0.09 & 1.41 & 3.48 & 11.10 \\
				30  & 546.50 & 0.09 & 1.44 & 0.09 & 1.42 & 3.50 & 11.14 \\
				40  & 554.09 & 0.09 & 1.47 & 0.09 & 1.42 & 3.51 & 11.21 \\
				50  & 558.60 & 0.09 & 1.50 & 0.09 & 1.44 & 3.53 & 11.27 \\
				60  & 563.04 & 0.09 & 1.53 & 0.09 & 1.45 & 3.54 & 11.33 \\
				70  & 566.40 & 0.09 & 1.53 & 0.09 & 1.45 & 3.53 & 11.48 \\
				80  & 569.69 & 0.09 & 1.53 & 0.09 & 1.45 & 3.51 & 11.56 \\
				90  & 571.10 & 0.09 & 1.54 & 0.09 & 1.46 & 3.52 & 11.59 \\
				100 & 572.25 & 0.09 & 1.54 & 0.09 & 1.47 & 3.52 & 11.67 \\
				\hline
			\end{tabular}
		\end{center}
	\end{table}
	\vspace{-0.7cm}
	\begin{table}[htb]
		\begin{center}
			\caption{Average results for Exp-2 under Gen-2 instance generation.}\label{table:exp2-gen2}
			\begin{tabular}{|c|c|c|c|c|c|c|c|}
				\hline
				& & \multicolumn{2}{|c|}{IP-FB}&\multicolumn{2}{|c|}{IP-CB}&\multicolumn{2}{|c|}{IP-LA} \\
				\cline{3-8}
				Scenarios & Cost & SHP & IP & SHP & IP & SHP & IP \\
				\hline
				20  & 624.21 & 0.09 & 1.26 & 0.09 & 1.41 & 3.49 & 12.4 \\
				30  & 626.00 & 0.09 & 1.29 & 0.09 & 1.42 & 3.50 & 12.8 \\
				40  & 627.58 & 0.09 & 1.33 & 0.09 & 1.42 & 3.49 & 13.2 \\
				50  & 628.70 & 0.09 & 1.34 & 0.09 & 1.43 & 3.50 & 13.7 \\
				60  & 629.78 & 0.09 & 1.34 & 0.09 & 1.41 & 3.51 & 12.6 \\
				70  & 630.10 & 0.09 & 1.35 & 0.09 & 1.43 & 3.51 & 13.3 \\
				80  & 630.45 & 0.09 & 1.33 & 0.09 & 1.42 & 3.49 & 13.4 \\
				90  & 631.00 & 0.09 & 1.36 & 0.09 & 1.44 & 3.50 & 13.8 \\
				100 & 631.32 & 0.09 & 1.37 & 0.09 & 1.45 & 3.49 & 13.9 \\
				\hline
			\end{tabular}
		\end{center}
	\end{table}

	Although the trend observed in Figure~\ref{fig:exp2-pp}-(b) is similar to that in Figures~\ref{fig:exp1-pp-gen1} and \ref{fig:exp1-pp-gen2}, a slightly different behavior can be observed in Figure~\ref{fig:exp2-pp}-(a). Specifically, IP-CB solves a larger fraction of instances within smaller $\tau$ factors, whereas IP-FB becomes the fastest method on an instance-wise basis as the $\tau$ factor increases.
	
	\begin{figure}[htb]
		\begin{center}
			\subfigure[Gen-1]{\includegraphics[width=0.23\textwidth]{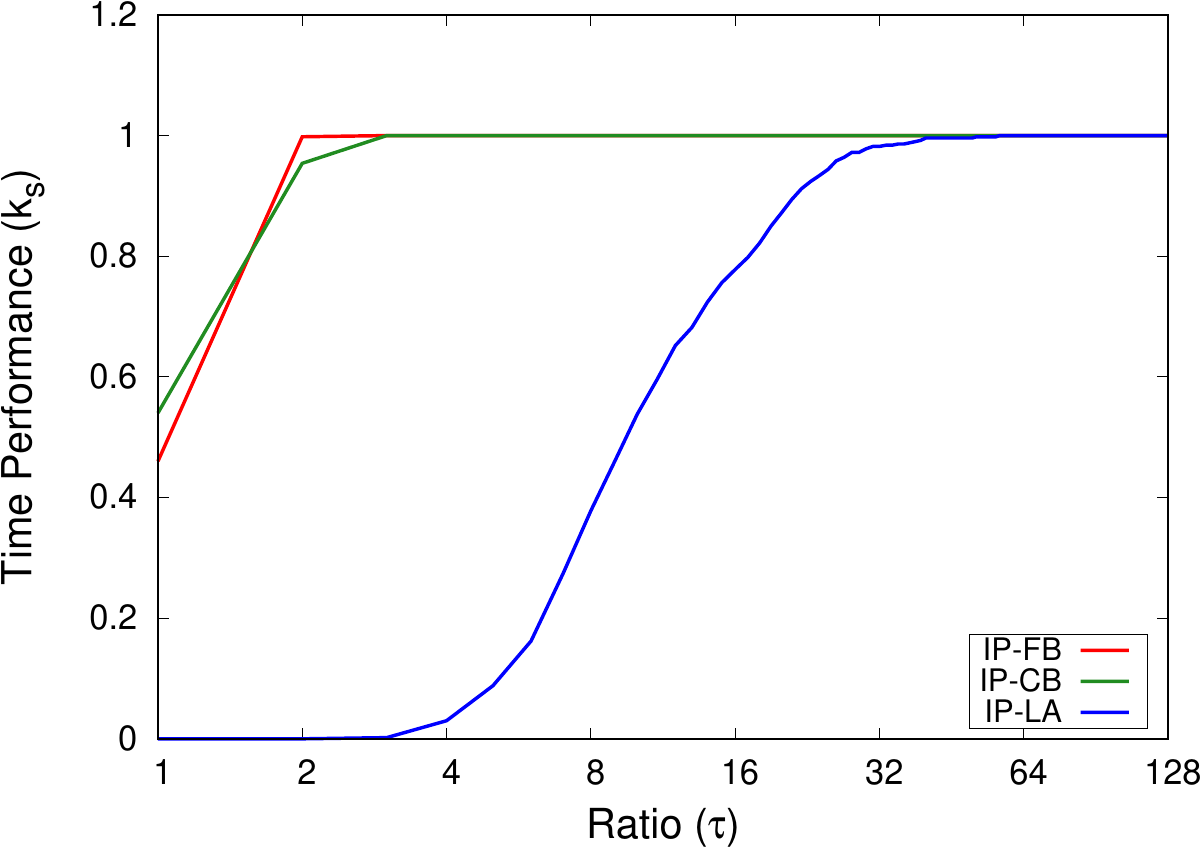}}
			\subfigure[Gen-2]{\includegraphics[width=0.23\textwidth]{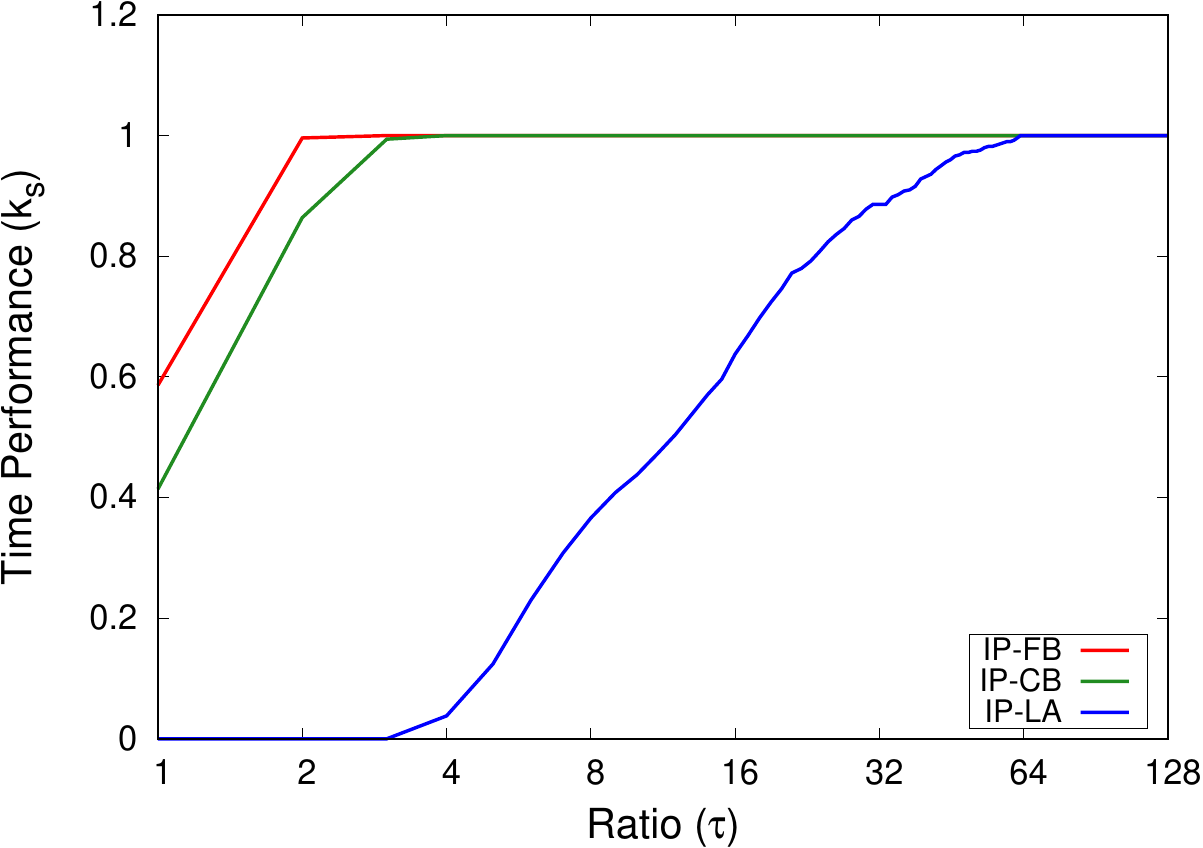}}
		\end{center}
		\caption{Performance profile of process time for solution methods in Exp-2.}\label{fig:exp2-pp}
	\end{figure}

	
	The results reported in Table~\ref{table:exp3-dens4-lprelax} present the average objective values (installation costs) obtained from the LP relaxations of IP-FB, IP-CB, and IP-LA. The column \textit{OPT} reports the optimal objective values of the corresponding integer programs. All values are rounded for clarity. Since the problem is formulated as a minimization problem, the LP relaxations yield lower objective values than the corresponding integer optimal solutions. Moreover, as all formulations share the same linear objective function, the method whose LP relaxation attains the highest objective value provides the tightest convex hull approximation of the integer feasible region.
	
	The results indicate that, for both Gen-1 and Gen-2, the proposed formulations yield tighter LP relaxations compared to IP-LA. In particular, IP-FB consistently dominates the other two methods. This property is advantageous, as LP relaxations are commonly used as heuristic approximations for solving computationally hard integer programs, and tighter relaxations generally lead to improved solution quality.
	\vspace{-0.2cm}
	\begin{table}[htb]
		\begin{center}
			\caption{Average objective value in Exp-3 for Dens=0.4}\label{table:exp3-dens4-lprelax}
			\begin{tabular}{|c|c|c|c|c|c|c|c|c|}
				\hline
				\multirow{2}{*}{Nodes} & \multicolumn{4}{|c|}{Gen-1} & \multicolumn{4}{|c|}{Gen-2} \\
				\cline{2-9}
				& OPT & FB & CB & LA & OPT & FB & CB & LA \\
				\hline
				8  & 973 & 964 & 960 & 948 & 987 & 981 & 977 & 969 \\
				9  & 924 & 913 & 907 & 889 & 952 & 942 & 933 & 906 \\
				10 & 871 & 852 & 834 & 789 & 898 & 875 & 861 & 820 \\
				11 & 836 & 806 & 785 & 725 & 873 & 847 & 825 & 764 \\
				12 & 807 & 772 & 746 & 676 & 844 & 810 & 783 & 708 \\
				13 & 790 & 745 & 713 & 632 & 837 & 802 & 769 & 682 \\
				14 & 784 & 732 & 699 & 606 & 837 & 802 & 769 & 682 \\
				15 & 746 & 695 & 661 & 569 & 788 & 742 & 705 & 603 \\
				16 & 746 & 689 & 651 & 550 & 782 & 731 & 691 & 584 \\
				\hline
			\end{tabular}
		\end{center}
	\end{table}
	\vspace{-0.2cm}
	
	In addition, Figure~\ref{fig:exp2-lprelax} reports the average ratio of the LP relaxation objective value to the optimal integer objective value. A higher ratio indicates a tighter feasible region. The results show that the LP relaxations of the proposed formulations achieve at least approximately 90 percent of the optimal objective value in this experiment, whereas the LP relaxation of IP-LA attains only about 75 percent.
	\vspace{-0.2cm}
	\begin{figure}[htb]
		\begin{center}
			\subfigure[Gen-1]{\includegraphics[width=0.23\textwidth]{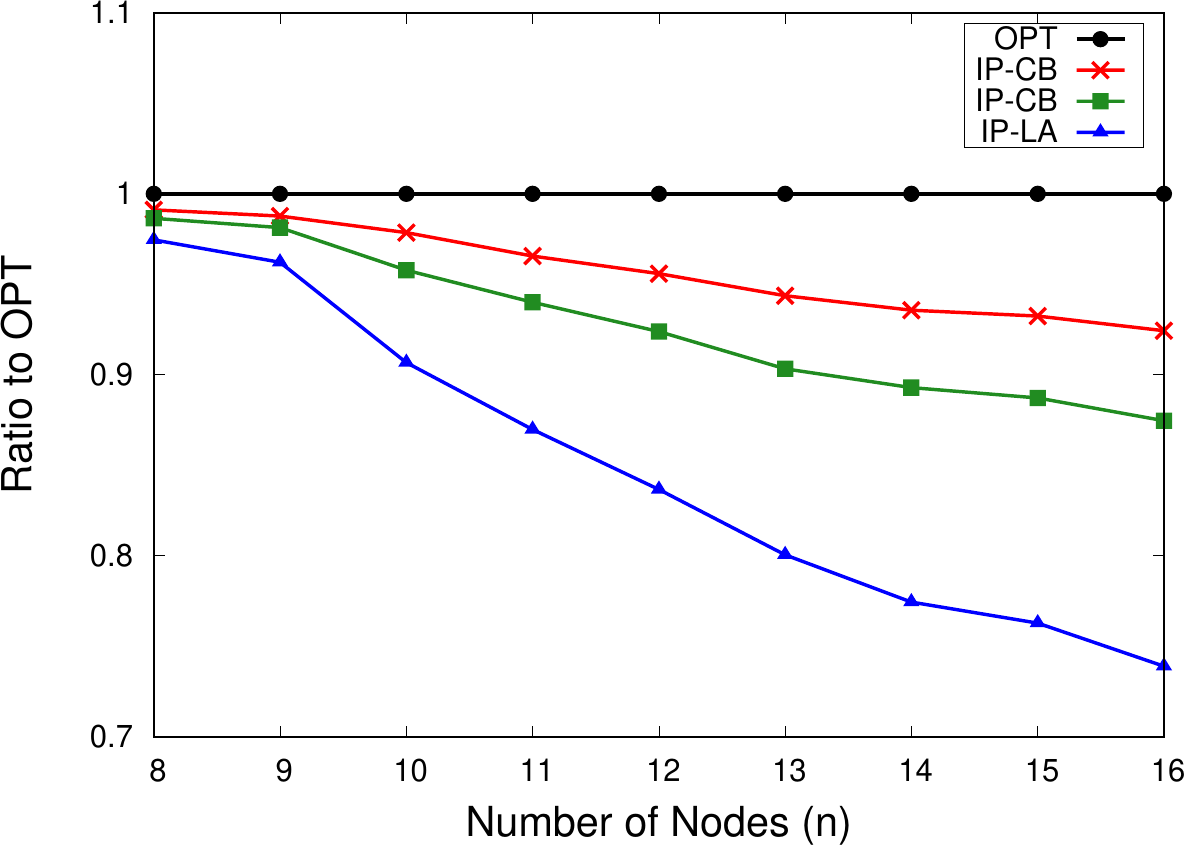}}
			\subfigure[Gen-2]{\includegraphics[width=0.23\textwidth]{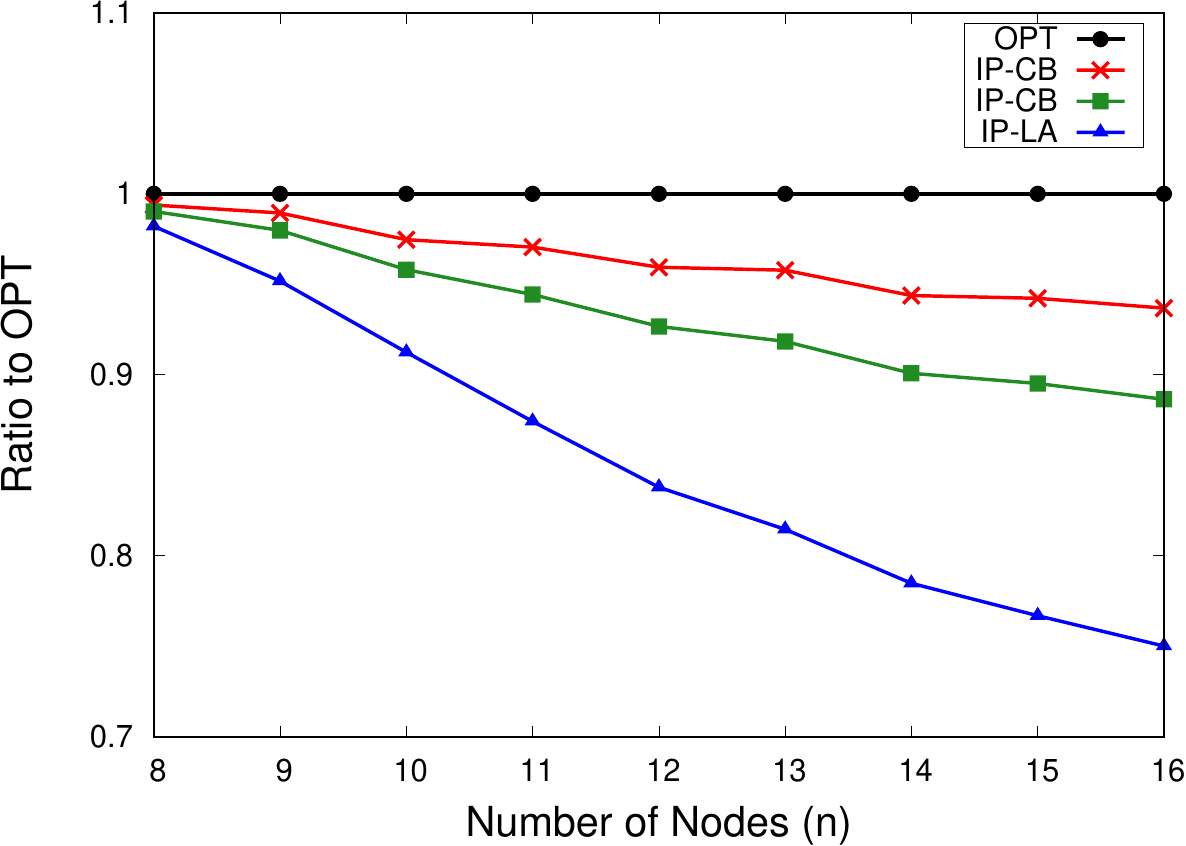}}
		\end{center}
		\caption{Average ratio of LP relaxations' obejctive value over the optimal solution of the corresponding IPs with $N=10$ and $dens=0.4$}\label{fig:exp2-lprelax}
	\end{figure}
	
	\subsection{Results of Large Instances}
	
	
	Exp-4 is designed to evaluate the effect of increasing the number of scenarios on large-scale instances solving by IT-FB. Results provided in Tables~\ref{table:exp3-gen1} and \ref{table:exp3-gen2} illustrate an acceptable solution time for large instances even though the solution time rises by increasing the number of scenarios. It must be noticed that instances with 100 nodes are categorized as large instances for robust version, while these are not considered as large instances in the related literature.
	
	However, unlike the results in the other experiments, here Gen-1 generates harder (in terms of solution time) instances for graphs with density of 0.4 and 0.8. In this case, the average of the most time-consuming size of instances remains less than one minute for $n=N=100$ with Dens=0.4. Moreover, increasing the density of the graph leads to less NDC pair of nodes and consequently easier instances.
	\vspace{-0.2cm}
	\begin{table}[htbp]
		\begin{center}
			\caption{Exp-4, results of Gen-1}\label{table:exp3-gen1}
			\begin{tabular}{|c|c|c|c|c|c|c|}
				\hline
				& \multicolumn{2}{|c|}{Dens=0.4}& \multicolumn{2}{|c|}{Dens=0.6}& \multicolumn{2}{|c|}{Dens=0.8} \\
				\cline{2-7}
				Scenaris & Cost & Time & Cost & Time & Cost & Time \\
				\hline
				10  & 1072.16 & 12.70 & 741.20 & 5.894 & 532.41 & 2.734 \\
				20  & 1108.85 & 17.96 & 778.48 & 8.904 & 556.29 & 3.938 \\
				30  & 1126.43 & 22.58 & 796.27 & 11.68 & 567.49 & 4.862 \\
				40  & 1140.35 & 27.40 & 809.14 & 14.93 & 575.45 & 5.882 \\
				50  & 1149.15 & 32.89 & 816.71 & 17.39 & 580.84 & 6.566 \\
				60  & 1157.42 & 39.82 & 824.51 & 20.72 & 584.95 & 7.458 \\
				70  & 1163.65 & 43.13 & 829.51 & 22.85 & 588.40 & 8.330 \\
				80  & 1168.38 & 49.45 & 833.38 & 25.73 & 591.19 & 8.991 \\
				90  & 1171.92 & 52.34 & 837.56 & 27.75 & 592.75 & 9.830 \\
				100 & 1175.03 & 56.93 & 841.06 & 33.37 & 594.85 & 10.45 \\
				\hline
			\end{tabular}
		\end{center}
	\end{table}
	\vspace{-0.7cm}
	\begin{table}[htbp]
		\begin{center}
			\caption{Exp-4, results of Gen-2}\label{table:exp3-gen2}
			\begin{tabular}{|c|c|c|c|c|c|c|}
				\hline
				& \multicolumn{2}{|c|}{Dens=0.4}& \multicolumn{2}{|c|}{Dens=0.6}& \multicolumn{2}{|c|}{Dens=0.8} \\
				\cline{2-7}
				Scenarios & Cost & Time & Cost & Time & Cost & Time \\
				\hline
				10  & 1333.23 & 15.26 & 947.78 & 13.89 & 640.64 & 2.381 \\
				20  & 1341.22 & 17.79 & 955.52 & 18.21 & 645.45 & 3.216 \\
				30  & 1344.28 & 19.37 & 959.48 & 21.12 & 648.64 & 3.731 \\
				40  & 1347.07 & 21.17 & 962.32 & 25.58 & 650.47 & 4.256 \\
				50  & 1349.09 & 23.26 & 963.96 & 28.41 & 651.32 & 4.822 \\
				60  & 1350.72 & 24.81 & 965.25 & 30.76 & 652.19 & 5.575 \\
				70  & 1351.38 & 27.31 & 966.54 & 33.95 & 652.99 & 6.187 \\
				80  & 1352.73 & 29.48 & 967.37 & 37.34 & 653.59 & 6.845 \\
				90  & 1353.37 & 30.51 & 968.01 & 39.53 & 654.15 & 7.461 \\
				100 & 1354.31 & 32.45 & 968.75 & 43.54 & 654.32 & 8.060 \\
				\hline
			\end{tabular}
		\end{center}
	\end{table}
	
	Exp-5 is devoted to the large instances to assess the impact of variation in the number of nodes. This experiment shows the highest solution times, in particular for Gen-2. In Tables~\ref{table:exp4-dens4}, \ref{table:exp4-dens6} and \ref{table:exp4-dens8} we added a new column to all considered graph densities, labeled as OPT. This column shows how many instances have been solved to optimality (out of 100 instances) within a 1200-second time limit.
	
	As shown in the results, all instance sizes are solved to optimality, except for the instances generated using Gen-2 with $n = 130$, $140$, and $150$. For these cases, only $96$, $92$, and $68$ instances, respectively, are solved to optimality. Accordingly, for these instance sets, we report only the average costs and solution times computed over the instances that reach optimality. For all other cases, since every instance is solved optimally, the reported averages are computed over the complete set of instances.
	\vspace{-0.2cm}
	\begin{table}[htbp]
		\begin{center}
			\caption{Exp-5, results for Dens=0.4}\label{table:exp4-dens4}
			\begin{tabular}{|r|c|r|c|c|r|c|}
				\hline
				& \multicolumn{3}{|c|}{Gen-1}& \multicolumn{3}{|c|}{Gen-2} \\
				\cline{2-7}
				Nodes & Cost & Time & OPT & Cost & Time  & OPT \\
				\hline
				50	& 1050 & 0.4	& 100 & 1203 & 0.6	 	& 100 \\
				60	& 1071 & 2.1	& 100 & 1273 & 2.5	 	& 100 \\
				70	& 1104 & 4.9	& 100 & 1311 & 5.2	 	& 100 \\
				80	& 1098 & 7.6	& 100 & 1334 & 8.7	 	& 100 \\
				90	& 1089 & 8.4	& 100 & 1341 & 12.6	 	& 100 \\
				100	& 1092 & 15.9	& 100 & 1341 & 17.7	 	& 100 \\
				110	& 1124 & 34.6	& 100 & 1345 & 32.3	 	& 100 \\
				120	& 1126 & 68.6	& 100 & 1347 & 59.3		& 100 \\
				130	& 1148 & 199.8	& 100 & 1369 & 228.9	& 96	\\
				140	& 1155 & 339.6	& 100 & 1406 & 531.5	& 92 \\
				150	& 1137 & 348.5	& 100 & 1466 & 888.1	& 68 \\
				\hline
			\end{tabular}
		\end{center}
	\end{table}
	\vspace{-0.7cm}
	\begin{table}[htbp]
		\begin{center}
			\caption{Exp-5, results for Dens=0.6}\label{table:exp4-dens6}
			\begin{tabular}{|r|c|r|c|c|r|c|}
				\hline
				& \multicolumn{3}{|c|}{Gen-1}& \multicolumn{3}{|c|}{Gen-2} \\
				\cline{2-7}
				Nodes & Cost & Time & OPT & Cost & Time  & OPT \\
				\hline
				50	& 662 & 0.1	& 100 & 795 & 0.2	& 100 \\
				60	& 772 & 1.8	& 100 & 869 & 1.5	& 100 \\
				70	& 772 & 3.1	& 100 & 916 & 3.6	& 100 \\
				80	& 773 & 4.8	& 100 & 950 & 7.1	& 100 \\
				90	& 763 & 5.1	& 100 & 953 & 11.1	& 100 \\
				100	& 773 & 7.3	& 100 & 955 & 18.1	& 100 \\
				110	& 791 & 14.3	& 100 & 956 & 28.4	& 100 \\
				120	& 781 & 16.1	& 100 & 956 & 41.7	& 100 \\
				130	& 783 & 29.5	& 100 & 957 & 52.4	& 100 \\
				140	& 781 & 33.5	& 100 & 957 & 66.1	& 100 \\
				150	& 790 & 88.1	& 100 & 959 & 86.1	& 100 \\
				\hline
			\end{tabular}
		\end{center}
	\end{table}
	\vspace{-0.7cm}
	\begin{table}[htbp]
		\begin{center}
			\caption{Exp-5, results for Dens=0.8}\label{table:exp4-dens8}
			\begin{tabular}{|r|c|r|c|c|r|c|}
				\hline
				& \multicolumn{3}{|c|}{Gen-1}& \multicolumn{3}{|c|}{Gen-2} \\
				\cline{2-7}
				Nodes & Cost & Time & OPT & Cost & Time  & OPT \\
				\hline
				50	& 507 & 0.1	& 100 & 580 & 0.1	& 100 \\
				60	& 504 & 0.1	& 100 & 580 & 0.1	& 100 \\
				70	& 532 & 0.4	& 100 & 583 & 0.1	& 100 \\
				80	& 523 & 0.3	& 100 & 587 & 0.3	& 100 \\
				90	& 510 & 2.1	& 100 & 613 & 1.2	& 100 \\
				100	& 536 & 0.9	& 100 & 645 & 2.9	& 100 \\
				110	& 530 & 3.9	& 100 & 688 & 5.5	& 100 \\
				120	& 516 & 3.3	& 100 & 721 & 9.8	& 100 \\
				130	& 608 & 12.7	& 100 & 735 & 13.3	& 100 \\
				140	& 614 & 17.9	& 100 & 759 & 20.4	& 100 \\
				150	& 612 & 24.9	& 100 & 754 & 25.3	& 100 \\
				\hline
			\end{tabular}
		\end{center}
	\end{table}

	\section{Conclusion}\label{sec:conclusion}
	
	Advance technologies in the field of networks and its service management has positively enhanced interaction between businesses. In particular, the emergence of digital platforms, cloud computing and data-driven technologies have significantly been shaping the world where nearly all parts of businesses depend on a reliable and resilient network infrastructure. Therefore, survivability is turned into one of the most important phases of modern network design. It essentially includes strategies to ensure that networks can survive during and after certain critical conditions.
	
	In this paper, we propose a new type of survivable network design problem, where the network is affected with two types of uncertainty sets. The first uncertainty set is in the form of a discrete uncertainty set and only affects the nodes of the network and thus the objective function of the problem. The second set is the discrete version of the budgeted interdiction uncertainty set and has made the problem structure uncertain by failing a subset of the edges. In addition, we formulate the problem mathematically using both the flow-based and cut-based idea and also introduce an iterative approach to effectively solve the large instances.
	
	In further research, it will be interesting to explore the field of RLP by considering the Min-Max Regret criterion under interval uncertainty for the edge lengths that combine the idea with the influence maximization algorithms to enhance the regenerator placement. Alternatively, it is interesting to define a budgeted uncertainty for the edge length and add machine learning techniques to approximate the possible placements of the regenerators.

	\appendices
	
	\section{Extension of FRE}\label{app:extension-fre}
	We now extend the FRE constraint to the case where $\Gamma \geq 2$. Consider a graph $G$ and its corresponding transformed graph $M$. The solutions of the FTRLP and the FTMCDSP on $G$ and $M$ are said to be equivalent if, for any node $v \in \overline{V}$, the number of neighbors of $v$ in $M$ that are equipped with regenerators is at least equal to the number of regenerator-equipped nodes that lie on all shortest paths in $G$ inducing edges incident to $v$ in $M$, plus an additional $\Gamma$ nodes. This condition ensures that the required level of fault tolerance is preserved in the transformed graph by guaranteeing sufficient redundancy in regenerator placement. To prove that this extension is valid, we can use the same technique used in proof of Theorem~\ref{thm:fre}.
	
	\section{Extension of Models}\label{app:extension}
	
	In this section, we show how our flow-based IP and iterative method can be extended to cases with $\Gamma \geq 2$. Based on the information provided in Section~\ref{sec:robust-counterpart}, it is clear that for cases with $\Gamma \geq 2$ there must exist $\Gamma+1$ edge-disjoint paths between any pair of nodes in the network. Now, we can show how to adjust the IP-FB and consequently the IT-FB to ensure $\Gamma+1$ between any pair of nodes. To this end, similar to the method shown in Section~\ref{subsec:flow-model}, assume we have paths $f^1, f^2, \ldots, f^{\Gamma+1}$. Therefore, for any $\gamma \in [\Gamma+1]$ we have the following constraints:
	\vspace{-0.2cm}
	
	\footnotesize
	\begin{align*}
	&\sum_{(p,j)\in E} f^{\gamma}_{pj,pq} - \sum_{(j,p)\in E} f^{\gamma}_{jp,pq} = x_p x_q & \forall p,q\in V \\
	&\sum_{(q,j)\in E} f^{\gamma}_{qj,pq} - \sum_{(j,q)\in E} f^{\gamma}_{jq,pq} = -x_p x_q & \forall p,q\in V \\
	&\sum_{(i,j)\in E} f^{\gamma}_{ij,pq} - \sum_{(j,i)\in E} f^{\gamma}_{ji,pq} = 0 & \forall p,q\in V,\; \forall i\in V\backslash(p,q)
	\end{align*}
	\normalsize
	Then, to ensure any node is connected to at least $\Gamma+1$ nodes of $L$, which can be formulated in the form of the following constraint.
	\vspace{-0.2cm}
	
	\footnotesize
	\begin{align*}
	& \quad \textstyle\sum_{j\in\mathcal{N}_i} x_j \ge \Gamma+1, && \forall i \in V
	\end{align*}
	\normalsize
	Moreover, it is necessary that paths $f^1, f^2, \ldots, f^{\Gamma+1}$ are edge-disjoint, thus we have:
	\vspace{-0.2cm}
	
	\footnotesize
	\begin{align*}
	& \quad f^1_{ij,pq} + f^1_{ij,pq} + \ldots f^{\Gamma+1}_{ij,pq} \le 1, && \forall p,q\in V,\; \forall (i,j)\in E
	\end{align*}
	\normalsize
	Finally, we add the next constraint to satisfy the FRE criterion:
	\vspace{-0.2cm}

	\footnotesize
	\begin{align*}
	& \quad \textstyle\sum_{j\in \mathcal{N}_i} x_j \geq \sum_{k\in\mathcal{N}_i^{\prime}} x_k +\Gamma  & \forall i\in \overline{V}
	\end{align*}
	\normalsize
	As for the IT-FB, we adjust the corresponding constraints in both master problem and subproblem of the iterative heuristic.

\end{document}